\documentclass[a4,12pt]{article}
\usepackage{graphicx,amsthm,fancyhdr,mathrsfs}
\usepackage{amsfonts}
\usepackage{amsmath}
\usepackage{multicol}
\usepackage{mathrsfs}
\usepackage{multirow}
\usepackage{amssymb}
\usepackage{url}
\usepackage{float}
   \usepackage[colorlinks=true,citecolor=blue,urlcolor=blue]{hyperref}
\usepackage{fancyhdr}
\usepackage{indentfirst}
\usepackage{enumerate}
\usepackage{enumitem}
\usepackage{amsthm}
\usepackage{color}
\usepackage{comment}
\usepackage{dsfont}
\usepackage{natbib}
\usepackage[misc]{ifsym}
\usepackage[toc,page]{appendix}
\usepackage{subfig}
\usepackage[top=25truemm,bottom=20truemm,left=20truemm,right=20truemm]{geometry}

\def\d{\mathrm{d}}

\def\parrow{\buildrel \mathrm p \over \rightarrow}
\def\darrow{\buildrel \d \over \rightarrow}

\newcommand{\var}{\mathrm{Var}}
\newcommand{\cov}{\mathrm{Cov}}

\newcommand{\id}{\mathds{1}}

\newcommand{\bzero}{\mathbf{0}}

\renewcommand{\ge}{\geqslant}
\renewcommand{\le}{\leqslant}

\renewcommand{\leq}{\leqslant}
\renewcommand{\epsilon}{\varepsilon}

\usepackage{array}
\newcommand{\PreserveBackslash}[1]{\let\temp=\\#1\let\\=\temp}
\newcolumntype{C}[1]{>{\PreserveBackslash\centering}p{#1}}
\newcolumntype{R}[1]{>{\PreserveBackslash\raggedleft}p{#1}}
\newcolumntype{L}[1]{>{\PreserveBackslash\raggedright}p{#1}}

\usepackage{booktabs,array}

\newcount\rowc

\makeatletter
\def\ttabular{%
\hbox\bgroup
\let\\\cr
\def\rulea{\ifnum\rowc=\@ne \hrule height 1.3pt \fi}
\def\ruleb{
\ifnum\rowc=1\hrule height 1.3pt \else
\ifnum\rowc=6\hrule height \heavyrulewidth
   \else \hrule height \lightrulewidth\fi\fi}
\valign\bgroup
\global\rowc\@ne
\rulea
\hbox to 10em{\strut \hfill##\hfill}%
\ruleb
&&%
\global\advance\rowc\@ne
\hbox to 10em{\strut\hfill##\hfill}%
\ruleb
\cr}
\def\endttabular{%
\crcr\egroup\egroup}

\theoremstyle{plain}
\newtheorem{Theorem}{Theorem}

\newtheorem{Lemma}{Lemma}

\newtheorem{Definition}{Definition}

\newtheorem{Remark}{Remark}

\usepackage{tikz}

\usepackage[compact]{titlesec}

\usepackage{color,xcolor,comment}
\usepackage{bm}

\makeatletter
\DeclareRobustCommand{\bsquare}{%
  \mathop{\vphantom{\sum}\mathpalette\bigstar@\relax}\slimits@
}

\newcommand{\bigstar@}[2]{%
  \vcenter{%
    \sbox\z@{$#1\sum$}%
    \hbox{\resizebox{.9\dimexpr\ht\z@+\dp\z@}{!}{$\m@th\dsquare$}}%
  }%
}
\makeatother

\usepackage[onehalfspacing]{setspace}

\begin{document}

\title{
Measuring and testing tail equivalence
}

\author{
Takaaki Koike\thanks{\protect\linespread{1}\protect\selectfont
Graduate School of Economics, Hitotsubashi University, 2-1, Naka, Kunitachi, Tokyo 186-8601, Japan.
Email: \texttt{takaaki.koike@r.hit-u.ac.jp}}
\and  Shogo Kato\thanks{\protect\linespread{1}\protect\selectfont Institute of Statistical Mathematics, 10-3 Midori-cho, Tachikawa, Tokyo 190-8562, Japan.
Email: \texttt{skato@ism.ac.jp}}
\and
Toshinao Yoshiba
\thanks{\protect\linespread{1}\protect\selectfont 
Graduate School of Management, Tokyo Metropolitan University, 18F 1-4-1 Marunouchi, Chiyoda-ku, Tokyo 100-0005, Japan.
Email:\texttt{tyoshiba@tmu.ac.jp}
}
}
\maketitle

\begin{abstract}
We call two copulas tail equivalent if their first-order approximations in the tail coincide.
As a special case, a copula is called tail symmetric if it is tail equivalent to the associated survival copula.
We propose a novel measure and statistical test for tail equivalence.
The proposed measure takes the value of zero if and only if the two copulas share a pair of tail order and tail order parameter in common. 
Moreover, taking the nature of these tail quantities into account, we design the proposed measure so that it takes a large value when tail orders are different, and a small value when tail order parameters are non-identical.
We derive asymptotic properties of the proposed measure, and then propose a novel statistical test for tail equivalence.
Performance of the proposed test is demonstrated in a series of simulation studies and empirical analyses of financial stock returns in the periods of the world financial crisis and the COVID-19 recession.
Our empirical analysis reveals non-identical tail behaviors in different pairs of stocks, different parts of tails, and the two periods of recessions.
\medskip
\\
\hspace{0mm} 
\end{abstract}

\newpage

\section{Introduction}\label{sec:introduction}
Dependence measurement between variables, especially among their tail events, is of great interest for modeling and managing risk in various fields such as finance, hydrology and seismology.
An important question regarding tail dependence is equality of two tail behaviors, which are typically compared, for example, between different markets, regions and times. 
For related studies, the reader is referred to~\cite{AngChen2002},~\cite{jondeau2016asymmetry},~\cite{bormann2020detecting},~\cite{can2024two} and references therein.

In this study, we capture (tail) dependence among multiple variables by \emph{copula}, which is a multivariate distribution with standard uniform marginal distributions~\citep{nelsen2006introduction}.
To quantify the degree of tail dependence, the \emph{tail dependence coefficient}~\citep[TDC,][]{sibuya1960bivariate} is a popular measure used in the literature.
We focus on comparing tail behaviors captured by the \emph{tail expansion} of a copula~\citep{HuaJoe2011}, which approximates the copula in the tail by a monimial.
More formally, a $d$-dimensional copula $C$ is said to have the \emph{(lower) tail order} $\kappa_{\operatorname{L}}(C)\ge 1$ if
\begin{align}\label{eq:def:lower:tail:order}
C(u,\dots,u)\sim u^{\kappa_{\operatorname{L}}(C)}\ell_{\operatorname{L}}(u;C),\quad u\downarrow 0,
\end{align}
for some slowly varying function $\ell_{\operatorname{L}}(\cdot;C)$ at $0$.
In~\eqref{eq:def:lower:tail:order}, the relationship ``$f(u)\sim g(u)$, $u\rightarrow a$'' for $f,g:\mathbb{R}\rightarrow \mathbb{R}$ and $a\in\mathbb{R}$ means $\lim_{u\rightarrow a}f(u)/g(u)=1$.
Moreover, a function $f: \mathbb{R}\rightarrow \mathbb{R}$ is called \emph{slowly varying} at $a\in[0,\infty]$ if $\lim_{x\rightarrow a}f(tx)/f(x)=1$ for any $t \in (0,\infty)$.
Denote by $\operatorname{SV}_a$ the set of all such functions.
If a copula $C$ admits~\eqref{eq:def:lower:tail:order} and the limit $\lim_{u\downarrow 0}\ell_{\operatorname{L}}(u;C)=\lambda_{\operatorname{L}}(C)$ exists, then $C$ is said to have the \emph{(lower) tail order parameter} $\lambda_{\operatorname{L}}(C)\in [0,\infty]$.
The copula $C$ is called \emph{(lower) tail dependent} when $\kappa_{\operatorname{L}}(C)=1$. 
In this case, $\lambda_{\operatorname{L}}(C)$ coincides with the TDC of $C$ when $d=2$.
The case $\kappa_{\operatorname{L}}(C)=d$ is referred to as \emph{tail independence}.
If $1< \kappa_{\operatorname{L}}(C)<d$, then $C$ is said to have \emph{intermediate tail dependence}.
Finally, the case $\kappa_{\operatorname{L}}(C)>d$ represents \emph{negative tail dependence}.
As such, the tail expansion~\eqref{eq:def:lower:tail:order} can capture weaker tail dependence that the TDC cannot.
Similarly to the lower case, a copula $C$ is said to have the \emph{upper tail order} $\kappa_{\operatorname{U}}(C)\ge 1$ if
\begin{align}\label{eq:def:upper:tail:order}
{\hat C}(u,\dots,u)\sim u^{\kappa_{\operatorname{U}}(C)}\ell_{\operatorname{U}}(u;C),\quad u\downarrow 0,
\end{align}
where $\hat C$ is a copula of $(1-U_1,\dots,1-U_d)$ with $(U_1,\dots,U_d)\sim C$, and $\ell_{\operatorname{U}}(\cdot;C)\in \operatorname{SV}_0$.
The \emph{upper tail order parameter} $\lambda_{\operatorname{U}}(C)\in [0,\infty]$ 
is the limit $\lim_{u\downarrow 0}\ell_{\operatorname{U}}(u;C)=\lambda_{\operatorname{U}}(C)$ provided it exists.

Let $C_1$ and $C_2$ be two $d$-dimensional copulas satisfying~\eqref{eq:def:lower:tail:order}  with (lower) tail orders $\kappa_1$ and $\kappa_2$ and (lower) tail order parameters $\lambda_1$ and $\lambda_2$, respectively.
We say that $C_1$ and $C_2$ are \emph{(lower) tail equivalent} if $\kappa_1=\kappa_2$ and $\lambda_1=\lambda_2$.
This concept includes an important special case of comparing upper and lower tails.
For a copula $C$ satisfying~\eqref{eq:def:lower:tail:order} and~\eqref{eq:def:upper:tail:order}, we say that $C$ is \emph{tail symmetric} if 
$C$ and $\hat C$ are tail equivalent, i.e.,
$\lambda_{\operatorname{L}}(C)=\lambda_{\operatorname{U}}(C)$ and $\kappa_{\operatorname{L}}(C)=\kappa_{\operatorname{U}}(C)$; otherwise we call $C$ \emph{tail asymmetric}.
It is important to detect tail asymmetry and check whether the asymmetry arises from the difference in tail orders or that in tail order parameters.
Such information is particularly beneficial for choosing a suitable copula fitting to data typically based on the trade-off between tractability and flexibility.
For example, it is well-known that prevalent Gaussian and $t$ copulas are tractable but tail symmetric, for which their skewed counterparts have been extensively studied in the literature; see~\cite{adcock2020selective} and \cite{azzalini2022overview} for comprehensive reviews.

It is not a trivial task to quantify the difference in two tails approximated by the tail expansion.
This is because tail order parameters are clearly comparable for tails with the identical tail order, and a tail with smaller tail order dominates another with larger tail order regardless of their tail order parameters.
Although this way of comparison could be conducted stepwise, it is beneficial to compare the difference of two tails in this way by a single number.
For this purpose, we propose a \emph{measure of tail equivalence} $\xi_{C_1,C_2,w}$ between two copulas $C_1$ and $C_2$ and for $w\in [0,1]$, defined as a limit of 
\begin{align*}
\xi_{C_1,C_2,w}(u,v)=(1-w)\,h_1\left(
\frac{\alpha_{C_1,C_2}(u)}{\log (1/u)}\right) + w\, h_2(\alpha_{C_1,C_2}(v)),
\quad u,v\in \left(0,1\right),
\end{align*}
as $u,v\downarrow 0$.
Here, the function $\alpha_{C_1,C_2}$ is defined by
\begin{align}\label{eq:measure:tail:asymmetry:intro:two:cop}
\alpha_{C_1,C_2}(u)=\log\left(\frac{C_2(u,\dots,u)}{C_1(u,\dots,u)}\right),\quad u\in \left(0,1\right),
\end{align}
and $h_l:\mathbb{R}\rightarrow [-1,1]$, $l=1,2$, are increasing functions satisfying certain properties introduced in~Definition~\ref{def:measure:tail:equivalence}.

When $C_1=C$ and $C_2=\hat C$, the measure $\alpha_{C,\hat C}(u)$ in~\eqref{eq:measure:tail:asymmetry:intro:two:cop} and its limit are introduced and studied in~\cite{Kato_etal2022}.
We transform these measures so that the proposed measure $\xi_{C_1,C_2,w}$ possesses various useful properties for measuring tail equivalence.
Firstly, it detects various equalities between tails depending on the choice of $w$.
For $w\in (0,1]$, it holds that $\xi_{C_1,C_2,w}=0$ if and only if $C_1$ and $C_2$ are tail equivalent, and for $w=0$, it holds that $\xi_{C_1,C_2,0}=0$ if and only if the tail orders of $C_1$ and $C_2$ coincide.
Secondly, the range of $\xi_{C_1,C_2,w}$ is $[-1,1]$, and the value of $\xi_{C_1,C_2,w}$ has clear interpretations.
A positive value of $\xi_{C_1,C_2,w}$ indicates that $C_2$ is more tail dependent than $C_1$, and a negative value implies the opposite.
More importantly, the measure takes a value on $[-w,w]$ when $C_1$ and $C_2$ have the same tail order.
Therefore, a small value of $\xi_{C,w}$ quantifies a difference in tail order parameters.
Meanwhile, $\xi_{C_1,C_2,w}$ takes a value in $[-1,1]\backslash[-w,w]$ when $C_1$ and $C_2$ have different tail orders
regardless of their tail order parameters.
Consequently, a difference in tail orders is represented by a large value of $\xi_{C_1,C_2,w}$.

These properties of $\xi_{C_1,C_2,w}$ are beneficial for statistically testing various hypotheses on tail orders and tail order parameters of $C_1$ and $C_2$ based on an estimator of $\xi_{C_1,C_2,w}$.
Given observations from $C_1$ and $C_2$ which are assumed to admit tail expansions, one can, for example, non-parametrically test H$_0$ : $C_1$ and $C_2$ are (lower) tail equivalent, rejecting which statistically supports some difference between two tails.
Note that equality between tail orders in disregard of tail order parameters can be tested based on existing estimators of tail order; see~\cite{falk2010laws},~\cite{embrechts2013modelling} and references therein.
For the tail dependent case, equality between tail order parameters can be tested based on existing estimators of the TDC; see, for example,~\cite{schmid2007nonparametric}.
We derive consistency and asymptotic normality of some estimators of $\xi_{C_1,C_2,w}$ and $\xi_{C_1,C_2,w}(u,v)$ for finite thresholds $u,v\in (0,1)$, based on which hypotheses on both of tail order and tail order parameter can be tested in a simple and natural manner.

We explore finite sample performance of these estimators and hypothesis tests in a series of simulation studies and empirical analyses on various tails of financial stock returns.
Our empirical studies reveal the following observations: 
(i) tail dependence structures, particularly the tail order parameters, are different between the two periods of recession, the world financial crisis in 2007--2009 and the COVID-19 recession;
(ii) asymmetry in the lower and upper tails of stock returns, also reported in~\cite{AngChen2002}, is observable in their tail orders;
(iii) different pairs of stock indices have different tail dependence structures, where the difference can lie not only in their tail order parameters but also in their tail orders.

The paper is organized as follows.
Section~\ref{sec:xi:properties} introduces our proposed measures and derives desirable properties for measuring tail equivalence.
Section~\ref{sec:statistics} presents estimators of the proposed measures and shows their consistency and asymptotic normality.
Based on these results, we explore statistical tests of tail equivalence.
In Section~\ref{sec:numerical}, we explore finite sample performance of the proposed measures and related hypothesis tests in a series of simulation studies and empirical analyses on stock returns.
Brief summary and concluding remarks are provided in Section~\ref{sec:conclusion}. 
We defer proofs, examples and some technical details in Appendices.

\section{The proposed measure of tail equivalence}\label{sec:xi:properties}

In this section we introduce the following measures to quantify the difference of the tail behavior between two copulas, and study their properties.

\begin{Definition}\label{def:measure:tail:equivalence}
Let $w\in [0,1]$ and $h_l:\mathbb{R}\rightarrow [-1,1]$, $l=1,2$, be functions satisfying the following properties:
\begin{enumerate}[label=(\roman*)]
    \item\label{h:prop:i} $h_l$ is increasing;
    \item\label{h:prop:ii} $h_l$ is strictly increasing and differentiable on $h_l^{-1}((-1,1))$;
    \item\label{h:prop:iii} $h_l(-x)=-h_l(x)$ for every $x \in \mathbb{R}$;
    \item\label{h:prop:iv} $\lim_{x\rightarrow \infty}h_l(x)=1$ and $\lim_{x\rightarrow -\infty}h_l(x)=-1$; and
    \item\label{h:prop:v} $h_l(0)=0$.
\end{enumerate}
For two $d$-dimensional copulas $C_1$ and $C_2$, the \emph{measure of tail equivalence} between $C_1$ and $C_2$ with finite thresholds $u,v\in(0,1)$ is defined by
\begin{align*}
\xi_{C_1,C_2,w}(u,v)=(1-w)\,h_1\left(
\frac{\alpha_{C_1,C_2}(u)}{\log (1/u)}\right) + w\, h_2(\alpha_{C_1,C_2}(v)),
\end{align*}
where 
\begin{align*}
\alpha_{C_1,C_2}(u)=\log\left(\frac{C_2(u,\dots,u)}{C_1(u,\dots,u)}\right),\quad u\in (0,1).
\end{align*}
Moreover, when the following limit exists: 
\begin{align*}
    \xi_{C_1,C_2,w} =\lim_{u,v\downarrow 0}\xi_{C_1,C_2,w}(u,v),
    \end{align*}
    we call it the \emph{measure of tail equivalence} between $C_1$ and $C_2$.
\end{Definition}

Note that the set of assumptions~\ref{h:prop:i}--\ref{h:prop:v} admits increasing functions with kinks as considered in~Remarks~\ref{rem:range:xi} and~\ref{rem:h2}.
The measure $\xi_{C_1,C_2,w}(u,v)$ compresses the two quantities $\alpha_{C_1,C_2}(u)/\log (1/u)$ and $\alpha_{C_1,C_2}(v)$ to a single value by respectively applying the auxiliary function $h_1$ and $h_2$ and then taking a weighted sum.
The conditions on $h_l$, $l=1,2$, are required for $\xi_{C_1,C_2,w}$ to satisfy desirable properties for quantifying tail equivalence, which will be presented in Theorem~\ref{thm:properties:xi} below.
Examples of such auxiliary functions include $(2/\pi)\cdot\operatorname{arctan}(\cdot)$ and $2\cdot (\Phi(\cdot)-1/2)$ where $\Phi$ is the cdf of the standard normal distribution.
In addition, we will also see that the parameter $w$ is the baseline value, above and below which the difference of two tail behaviors lies either in tail orders or in tail order parameters.
Let $\operatorname{sign}(x)=\id(x>0) - \id(x<0)$, where $\id$
 is the indicator function.

\begin{Theorem}\label{thm:properties:xi}
Let $C_1$ and $C_2$ be $d$-dimensional copulas satisfying the following assumptions:
\begin{enumerate}[label=(A.\arabic*)]
    \item\label{item:A1} $C_l$ admits a tail order $\kappa_l$ and a tail order parameter $\lambda_l$ for $l=1,2$;
\item\label{item:A2} if $\kappa_1=\kappa_2$, then $(\lambda_1,\lambda_2)\neq (0,0)$ and $(\lambda_1,\lambda_2)\neq (\infty,\infty)$; and 
\item\label{item:A3} the limit
$\xi_{C_1,C_2,w} =\lim_{u,v\downarrow 0}\xi_{C_1,C_2,w}(u,v)$ exists.
\end{enumerate}
Then
\begin{align}\label{eq:xi:limit}
    \xi_{C_1,C_2,w}=
    \begin{cases}
    (1-w)\,h_1(\kappa_1 - \kappa_2)+     w \,h_2\left(\log\left(\lambda_2/\lambda_1\right)\right),& \text{ when }\kappa_1=\kappa_2,\\
       (1-w)\,h_1(\kappa_1 - \kappa_2) +w\, \operatorname{sign}(\kappa_1 - \kappa_2),&\text{ when }\kappa_1\neq \kappa_2.\\
    \end{cases}
\end{align}
Therefore, $\xi_{C_1,C_2,w}$ satisfies the following properties.
\begin{enumerate}[label=(\alph*)]
\item\label{xi:prop:range} (Range) For $w\in [0,1]$, it holds that:
\begin{align*}
    \xi_{C_1,C_2,w}\in
        \begin{cases}
    [-1,-w)\\
    [-w,0) \\
    (0,w] \\
    (w,1] \\
\end{cases}
\quad \Longleftrightarrow \quad 
    \begin{cases}
    \kappa_1<\kappa_2,\\
    \kappa_1=\kappa_2\text{ and }\lambda_1>\lambda_2,\\
   \kappa_1=\kappa_2\text{ and }\lambda_1<\lambda_2,\\
    \kappa_1>\kappa_2,\\
    \end{cases}
\end{align*}
where we interpret $[-1,-1)$, $(1,1]$, $[0,0)$ and $(0,0]$ as empty sets.
\item \label{xi:prop:symmetry} (Symmetry)
$\xi_{C_1,C_2,w}=-\xi_{C_2,C_1,w}$ for every $w\in [0,1]$.
\item\label{xi:prop:tail:equivalence}  (Tail equivalence)
\begin{align*}
    \xi_{C_1,C_2,w}=0 \quad \Longleftrightarrow \quad \begin{cases}
        \kappa_1=\kappa_2\text{ and }\lambda_1=\lambda_2, & \text{ when } w\in (0,1],\\
        \kappa_1=\kappa_2,& \text{ when } w=0.\\
    \end{cases}
\end{align*}
\item\label{xi:prop:monotonicity}  (Monotonicity)
Let $\tilde C_1$ and $\tilde C_2$ be $d$-dimensional copulas satisfying~\ref{item:A1},~\ref{item:A2} and~\ref{item:A3} with tail orders $\tilde \kappa_1,\tilde \kappa_2$ and tail order parameters $\tilde \lambda_1,\tilde \lambda_2$, respectively. 
Let $\Delta \kappa=\kappa_1-\kappa_2$, $\Delta \tilde \kappa=\tilde \kappa_1-\tilde \kappa_2$, $\Delta \lambda=\lambda_2-\lambda_1$ and $\Delta \tilde \lambda=\tilde \lambda_2-\tilde \lambda_1$.
Then, for $w\in[0,1)$, it holds that
\begin{align*}
    \xi_{C_1,C_2,w}< \xi_{\tilde C_1,\tilde C_2,w}\text{ for }\xi_{C_1,C_2,w}, \xi_{\tilde C_1,\tilde C_2,w}\in [-1,1]\backslash[-w,w] \quad \Longleftrightarrow\quad \Delta\kappa < \Delta\tilde\kappa. 
\end{align*}
For $w \in (0,1]$, it holds that
\begin{align*}
    \xi_{C_1,C_2,w}< \xi_{\tilde C_1,\tilde C_2,w}\text{ for }\xi_{C_1,C_2,w}, \xi_{\tilde C_1,\tilde C_2,w}\in (-w,w) \quad \Longleftrightarrow\quad
    \Delta\kappa = \Delta\tilde\kappa\text{ and } \Delta\lambda < \Delta\tilde\lambda.
\end{align*}
\end{enumerate}
\end{Theorem}

Note that larger tail order parameter and smaller tail order (with finite tail order parameter) imply stronger tail dependence.

\begin{Remark}\label{rem:range:xi}
Assuming that the underlying copulas $C_l$, $l=1,2$, satisfy $u^d \le C_l(u,\dots,u)$ for sufficiently small $u$, it holds that $\kappa_l \in [1,d]$ together with the Fr\'echet-Hoeffding bound $C_l(u,\dots,u)\le u$ for all $u\in [0,1]$.
In such a situation, one can take 
\begin{align*}
    h_1(x)=\begin{cases}
        -1, & \text{ if } x < -(d-1),\\
        x/(d-1),  & \text{ if }  -(d-1)\le x \le d-1,\\
        1,  & \text{ if } x > (d-1),\\
    \end{cases}
\end{align*} 
under which it holds, for $w\in [0,1)$, that $\xi_{C_1,C_2,w}=1$ if and only if $(\kappa_1,\kappa_2)=(d,1)$ and that $\xi_{C_1,C_2,w}=-1$ if and only if $(\kappa_1,\kappa_2)=(1,d)$.
\end{Remark}

\begin{Remark}\label{rem:infinite:tail:order:parameter}
We exclude the two cases $(\lambda_1,\lambda_2)\neq (0,0)$ and $(\lambda_1,\lambda_2)\neq (\infty,\infty)$ in Theorem~\ref{thm:properties:xi} for two reasons.
One is that, in such cases, $\alpha_{C_1,C_2}(v)$ may not converge to $0$ even if $\lambda_1=\lambda_2$.
For example, we have that $\lim_{v\downarrow 0}\alpha_{C_1,C_2}(v)=\infty$ when $\ell_1(v)=a_1 (-\log v )^{b_1}$ and $\ell_2(v)=a_2 (-\log v)^{b_2}$ for $a_1,a_2>0$ and $b_1,b_2<0$ with $b_2>b_1$.
Another reason is that the case $\lambda_1=\lambda_2=\infty$ requires special treatment since, in this case, higher tail order indicates stronger tail dependence, which is opposite to the other cases~\citep{HuaJoe2011}.
\end{Remark}

\begin{Remark}
If $w=0$, then the measure reduces to $\xi_{C_1,C_2,w}=h_1(\kappa_1-\kappa_2)$, which is completely irrelevant to tail order parameters.
On the other hand, even if $w=1$, the measure $\xi_{C_1,C_2,w}$ is affected by tail orders since  $\xi_{C_1,C_2,w}\in(-1,1)$ only when $\kappa_1=\kappa_2$.
Moreover, $\xi_{C_1,C_2,w}\in\{-1,1\}$ is attained not only when  $\kappa_1\neq \kappa_2$ but also when $\kappa_1=\kappa_2$ and
$\log(\lambda_2/\lambda_1)$ is $\infty$ or $-\infty$.
\end{Remark}

\section{Statistical inference}\label{sec:statistics}

This section studies statistical estimation and hypothesis testing on the proposed measures $\xi_{C_1,C_2,w}$ and $\xi_{C_1,C_2,w}(u,v)$ for finite $u,v\in(0,1)$.
We consider the following situation.
\begin{enumerate}[label=(A.\arabic*)]\setcounter{enumi}{3}
    \item\label{item:A4:iid}
    For $i=1,2,\dots$, $(\mathbf{U}_i^{(1)}, \mathbf{U}_i^{(2)})$ follows the $2d$-dimensional copula $C$ with $\mathbf{U}_i^{(l)}=(U_{i1}^{(l)},\dots,U_{id}^{(l)})\sim C_l$ for $l=1,2$.
Moreover, $\mathbf{U}^{(k)}_i$ and $\mathbf{U}^{(l)}_j$ are independent for every $i,j=1,2,\dots$ with $i \neq j$ and $k,l=1,2$.
Finally, $C_1$ and $C_2$ satisfy~\ref{item:A1},~\ref{item:A2} and~\ref{item:A3} in Theorem~\ref{thm:properties:xi}.
\end{enumerate}

For the later convenience, write $M_i^{(l)}=\max(U_{i1}^{(l)},\dots,U_{id}^{(l)})\sim F_l$, $l=1,2$, and $(M_i^{(1)},M_i^{(2)})\sim H$ for which it holds that $F_l(u)=C_l(u,\dots,u)$, $u\in [0,1]$, and that $H(u,v)=C(\mathbf{u},\mathbf{v})$ with $\mathbf{u}=(u,\ldots,u) \in [0,1]^d$ and $\mathbf{v}=(v,\ldots,v) \in [0,1]^d$. 
The (modified) empirical cdfs of $H$, $F_1$, $F_2$ are
$\hat H(u,v)=(1/n)+(1/n)\, \sum_{i=1}^n \id (
M_i^{(1)}\le u,M_i^{(2)}\le u)$, 
$\hat F_1(u)=\hat H(u,1)$ and $\hat F_2(v)=\hat H(1,v)$,  $u,v\in[0,1]$, respectively.
We consider the modified versions for the ratio $\hat F_2(u)/\hat F_1(u)$ to be well-defined for every $u\in[0,1]$.

\subsection{Statistical inference on $\xi_{C_1,C_2,w}(u,v)$}\label{sec:inference:finite:u:v}

We first consider statistical inference for the case when the thresholds $u,v\in (0,1)$ are fixed, in which we estimate $\xi_w(u,v)$ instead of $\xi_w$.

In the following method, the measure $\alpha_{C_1,C_2}(u)$ is replaced by its empirical version. 
Adopting this empirical measure, we define an empirical measure of $\xi_{C_1,C_2,w} (u,v)$ and study its asymptotic properties.

\begin{Definition}\label{def:estimator}
Define an empirical version of the measure $\alpha_{C_1,C_2}(u)$ by
\begin{equation}
\hat{\alpha}(u) = \log \left( \frac{\hat F_2(u)}{\hat F_1(u)} \right), 
\quad u\in(0,1).
\label{eq:nu_hat}
\end{equation}
Then an empirical version of the proposed measure $\xi_{C_1,C_2,w}$ is defined by
\begin{align}
	\hat{\xi}_{w}(u,v)=(1-w)\,h_1\left(
	\frac{\hat{\alpha}(u)}{\log (1/u)}\right) + w\, h_2(\hat{\alpha}(v)),
	\quad u,v\in (0,1),\ w\in [0,1]. \label{eq:xi_hat}
\end{align}
\end{Definition}

\begin{Remark}
One of the proposed measures (\ref{eq:nu_hat}) can be considered a multivariate extension of the bivariate and trivariate empirical measures of \cite{Kato_etal2022}.
The measure (\ref{eq:nu_hat}) reduces to the empirical measure of \cite{Kato_etal2022} when $d=2$ and $C_2$ is the survival copula of $C_1$.
Note also that our empirical measure compares the lower tail probabilities of two different copulas, while the empirical measure of \cite{Kato_etal2022} is restricted to compare lower and upper tail probabilities of a single copula.
\end{Remark}

In order to discuss the asymptotic results of the empirical measure~\eqref{eq:xi_hat}, we assume $u=v$ for simplicity although our proof is applicable to more general case $u\neq v$.
For brevity, we write ${\xi}_{C_1,C_2,w}(u)={\xi}_{C_1,C_2,w}(u,u)$ and $\hat{\xi}_{w}(u)=\hat{\xi}_{w}(u,u)$ for $u\in (0,1)$.

We now establish the following asymptotic result.
Its proof is straightforward from the law of large numbers (LLN) and the continuous mapping theorem (CMT).

\begin{Theorem}\label{thm:consistency}
For $u,v \in (0,1)$, the proposed empirical measures $\hat{\alpha}(u)$ and $\hat{\xi}_{w}(u)$ are consistent estimators of $\alpha_{C_1,C_2}(u)$ and $\xi_{C_1,C_2,w}(u,v)$, respectively.
\end{Theorem}

In addition a related sequence converges weakly to a Gaussian process.
We write $x\vee y = \max(x,y)$ and $x\wedge y = \min(x,y)$ for $x,y\in\mathbb{R}$.

\begin{Theorem} \label{thm:asymptotics}
	Define
	$$
	\mathbb{X}_n(u_i) = \sqrt{n} \left\{ 
	\hat{\xi}_{w}(u_i) - \xi_{C_1,C_2,w}(u_i) \right\}, \quad u_i \in ( 0,1), \quad i=1,\ldots,m.
	$$
In addition to Assumption~\ref{item:A4:iid}, assume that $\alpha_{C_1,C_2}(u_j)/\log (1/u_j)\in h_1^{-1}((-1,1))$ and $\alpha_{C_1,C_2}(u_j)\in h_2^{-1}((-1,1))$, respectively, for $j=1,\dots,m$.
	Then, as $n \rightarrow \infty$, $(\mathbb{X}_n(u_1) , \ldots, \mathbb{X}_n(u_m) )$ converges weakly to the multivariate normal distribution with mean vector $\bzero_m$ and the variance matrix $(\sigma_{w}(u_i,u_j))_{1\leq i,j \leq m}$, where 
 \begin{align*}
 \sigma_w (u_i,u_j)= a_w(u_i)\,a_w(u_j)\,A(u_i,u_j),\quad i,j=1,\dots,m,   
 \end{align*}
 with
\begin{align*}
    a_w (u_j) = \frac{1-w}{\log (1/u_j)}\,h_1'\left(
\frac{\alpha_{C_1,C_2}(u_j)}{\log (1/u_j)}\right)  + w\, h_2'(\alpha_{C_1,C_2}(u_j)),
\end{align*}
and
\begin{align*}
    A(u_i,u_j)=
\frac{F_1(u_i\wedge u_j)}{F_1(u_i)F_1(u_j)}+
\frac{F_2(u_i\wedge u_j)}{F_2(u_i)F_2(u_j)}
-\frac{H(u_i,u_j)}{F_1(u_i)F_2(u_j)}
-\frac{H(u_j,u_i)}{F_1(u_j)F_2(u_i)}.
\end{align*}
\end{Theorem}

\begin{Remark}
Let
$\mathbb{N}_n(u_i) = \sqrt{n} \left\{ 
	\hat{\alpha}(u_i) - \alpha_{C_1,C_2}(u_i) \right\}$, $u_i \in (0,1)$, $i=1,\ldots,m$.
 Then it is straightforward to prove that $(\mathbb{N}_n(u_1) , \ldots , \mathbb{N}_n(u_m))$ weakly converges to the multivariate normal distribution with mean vector $\bzero_m$ and the variance matrix $(A(u_i,u_j))_{1\leq i,j \leq m}$.
Note that this covariance matrix is related to the covariance matrix derived from the sequence of the empirical measures discussed in Theorem 4 of \cite{Kato_etal2022}.
\end{Remark}

We now consider the case when $m=1$, and let $ \sigma_w^2(u)={a}_w^2(u)\,{A}(u)$ with $A(u)=A(u,u)$.
Denote by $\hat \sigma_w^2(u)={\hat a}_w^2(u)\,{\hat A}(u)$ a plug-in estimator of $\sigma_w^2(u)$ with
\begin{align*}
    {\hat a}_w (u) = \frac{1-w}{\log (1/u)}\,h_1'\left(
\frac{\hat{\alpha}(u)}{\log (1/u)}\right)  + w\, h_2'(\hat{\alpha}(u))
\end{align*}
and
\begin{align*}
\hat A(u)=
\frac{1}{\hat F_1(u)}+
\frac{1}{\hat F_2(u)}-2\,
\frac{\hat H(u,u)}{\hat F_1(u)\hat F_2(u)}.
\end{align*}
Then we have that
$\hat \sigma_w(u)\parrow\sigma_w(u)$, and thus Theorem~\ref{thm:asymptotics} and CMT yield
\begin{align}\label{eq:asymptotics:T_n}
    \frac{\sqrt{n} \left\{ 
	\hat{\xi}_{w}(u) - \xi_{C_1,C_2,w}(u) \right\}}{\hat \sigma_w (u)} \darrow \mathcal N(0,1),\quad \text{ as }n\rightarrow\infty.
\end{align}

Based on~\eqref{eq:asymptotics:T_n}, one can construct confidence intervals and conduct hypothesis tests on $\xi_{C_1,C_2,w}(u)$.
For example, under the null hypothesis H$_0$ : $\xi_{C_1,C_2,w}(u)=0$, the statistic $T_n = \sqrt{n}\hat \xi_w(u)/\hat \sigma_w(u)$ asymptotically follows the standard normal distribution, and thus we reject H$_0$ based on the (asymptotic) p-value $2\min(\Phi(T_n),1-\Phi(T_n))$ for the two-tailed test against H$_1$ : $\xi_{C_1,C_2,w}(u)\neq 0$, 
$1-\Phi(T_n)$ for the right-tailed test against H$_1$ : $\xi_{C_1,C_2,w}(u) >0$, and $\Phi(T_n)$ for the left-tailed test against H$_1$ : $\xi_{C_1,C_2,w}(u)<0$.

A small threshold $u\in(0,1)$ is taken for comparing two tails.
In this situation, there is a typical trade-off to choose a smaller threshold for better approximation of the tail behavior and to choose a larger threshold for increasing the sample size.
To deal with this trade-off, the estimate $\hat \xi_w(u)$ and its confidence interval are typically computed for various thresholds to visually diagnose stability of the results; see~Section~\ref{sec:numerical} for illustration.

\subsection{Statistical inference on $\xi_{C_1,C_2,w}$}\label{sec:inference:limit:xi}

In this section, we consider statistical inference on the limit $\xi_{C_1,C_2,w}$.
In addition to~\ref{item:A4:iid}, we further assume that
\begin{enumerate}[label=(A.\arabic*)]\setcounter{enumi}{4}
    \item\label{item:A5:independence} $\mathbf{U}_i^{(1)}$ is independent of $\mathbf{U}_i^{(2)}$ for $i=1,2,\dots$, and
    \item\label{item:A6:finite:lambda} $C_l$ admits a finite positive tail order parameter $
    \lambda_l\in (0,\infty)$ for $l=1,2$.
\end{enumerate}

For $w\in [0,1]$, we consider the following form of  estimator:
\begin{align}\label{eq:xi:est}
   \hat \xi_w= (1-w)\,h_1\left( \hat \kappa_{1}-\hat \kappa_{2}
    \right) + w  \,h_2\left(
    \hat \alpha(v_n)
    \right),
\end{align}
where $\hat \kappa_l$ is an estimator of $\kappa_l$ for $l=1,2$, and $v_n\in(0,1)$ is such that $v_n\rightarrow 0$ as $n\rightarrow \infty$.
In~\eqref{eq:xi:limit}, the first term $(1-w)\,h_1(\kappa_1-\kappa_2)$ remains unchanged regardless of the relationship between $\kappa_1$ and $\kappa_2$.
Moreover, estimators of $\kappa_l$, $l=1,2$, and their asymptotic properties have been extensively studied in the literature since $\kappa_l$ is the \emph{tail index} of the random variable $X_i^{(l)}=1/M_i^{(l)}$, which satisfies $\mathbb{P}(X_i^{(l)}>x)=\mathbb{P}(M_i^{(l)}<1/x)=F_l(1/x)\sim x^{-\kappa_l} \ell_l(x)$, $x\rightarrow \infty$, for some $\ell_l\in\operatorname{SV}_{\infty}$.
The reader is referred to~\cite{embrechts2013modelling} and references therein.
Hence the asymptotic behavior of the first term in~\eqref{eq:xi:est} can be derived from the existing results, and the key part is the second term to describe the asymptotic behavior of $\hat \xi_w$.
Considering this situation, we do not specify estimators of $\kappa_l$, $l=1,2$, in our estimator~\eqref{eq:xi:est} and adopt assumptions~\ref{item:A:tail:index:consistency} and~\ref{item:A:tail:index:clt} to show consistency and asymptotic normality in the following theorems.

\begin{Theorem}\label{thm:consistency:hat:xi}
 In addition to \ref{item:A4:iid},~\ref{item:A5:independence} and~\ref{item:A6:finite:lambda}, assume that, for $l=1,2$,
 \begin{enumerate}[label=(A.\arabic*)]\setcounter{enumi}{6}
 \item\label{item:A:tail:index:consistency} $\hat \kappa_l \parrow \kappa_l$ as $n\rightarrow \infty$; and 
\item\label{item:A:v:n:consistency} $v_n\rightarrow 0$ and $n v_n^{\kappa_l}\rightarrow \infty$ as $n\rightarrow \infty$.
\end{enumerate}
  Then, for $w\in[0,1]$, it holds that 
  $\hat \xi_w\parrow \xi_w$ as $n\rightarrow\infty$.
\end{Theorem}

\begin{Theorem}\label{thm:clt:hat:xi}
Let $\kappa_{+}=\kappa_1\vee\kappa_2$.
    In addition to \ref{item:A4:iid},~\ref{item:A5:independence},~\ref{item:A6:finite:lambda} and~\ref{item:A:tail:index:consistency}, assume that
    there exist $m_n,\tilde m_n\in\mathbb{R}$ and $v_n \in (0,1)$ satisfying the following assumptions:
 \begin{enumerate}[label=(A.\arabic*)]\setcounter{enumi}{8}
 \item\label{item:mn:un:lim}  $m_n,\tilde m_n\rightarrow \infty$ and $v_n\rightarrow 0$ as $n\rightarrow\infty$;
 \item\label{item:A:tail:index:clt} $\sqrt{\tilde m_n}\{h_1(\hat \kappa_1-\hat \kappa_2)-h_1(\kappa_1-\kappa_2)\}\darrow \mathcal N(0,\tilde \sigma^2)$ as  $n\rightarrow\infty$ for some $\tilde \sigma>0$;
\item\label{item:A:tau} $m_n/(n v_n^{\kappa_{+}})= \tau$ for some $\tau > 0$;
\item\label{item:A:F:l} $\sqrt{m_n}\{
F_l(v_n)-v_n^{\kappa_l}\ell_{\operatorname L}(v_n;C_l)
\}\rightarrow 0$ as $n\rightarrow\infty$ for $l=1,2$; 
\item\label{item:A:h:2} 
$x^\ast =\operatorname{sup}\{x\in \mathbb{R}: h_2(x)\in(0,1)\}\in(0,\infty)$; and
\item\label{item:A:lambda:diff} 
$-x^\ast < \log(\lambda_2/\lambda_1)<x^\ast$.
\end{enumerate}
\begin{enumerate}[label=(\Roman*)]
\item\label{item:thm:case:I}
Suppose that $m_n/\tilde m_n\rightarrow \infty$ or $w=0$. Then
\begin{align}\label{eq:clt:case:I}
    \sqrt{\tilde m_n}\left(\hat \xi_w - \xi_w\right)\darrow \mathcal N(0,(1-w)^2\,\tilde \sigma^2)\quad \text{ as }n\rightarrow \infty.
\end{align}
    \item\label{item:thm:case:II} Suppose that $\tilde m_n/m_n\rightarrow \infty$ or $w=1$. If $\kappa_1=\kappa_2$, then
\begin{align}\label{eq:clt:case:II}
    \sqrt{m_n}\left(\hat \xi_w - \xi_w\right)\darrow  \mathcal N(0,w^2\,\sigma^2)\quad \text{ as }n\rightarrow \infty,
\end{align}
where
\begin{align}\label{eq:sigma:clt:h2}
    \sigma^2=
        \tau\left\{h_2'\left(
\log\frac{\lambda_2}{\lambda_1}
\right)\right\}^2
\left\{
\frac{
1
}{\lambda_1}+\frac{
1
}{\lambda_2}
\right\}.
\end{align}
\end{enumerate}
  \end{Theorem}

Note that~\ref{item:A:v:n:consistency} is implied by~\ref{item:mn:un:lim} and~\ref{item:A:tau}.
Moreover,~\ref{item:A:lambda:diff} ensures that $h_2'(\log(\lambda_2/\lambda_1))>0$, and thus $\sigma^2>0$ in~\eqref{eq:sigma:clt:h2}.
We study the cases $m_n/\tilde m_n\rightarrow \infty$ and  $\tilde m_n/m_n\rightarrow \infty$ because particularly the former case is typically satisfied for standard choices of copulas and an estimator of tail indices, which we will discuss in the end of this section and in Appendix~\ref{sec:example}.
For the latter case with $\kappa_1\neq \kappa_2$, we are not able to find a suitable sequence $(m_n,\tilde m_n)$ for some asymptotic normality to hold.

\begin{Remark}\label{rem:h2}
    For \ref{item:A:h:2}, the threshold $x^\ast$ can be interpreted as the maximum difference between $\log \lambda_2$ and $\log \lambda_1$ which can be measured by $h_2$.
This assumption is satisfied when, for example, 
\begin{align*}
    h_2(x)=\begin{cases}
        -1, & \text{ if } x < -x^\ast,\\
        x/x^\ast,  & \text{ if }  -x^\ast\le x \le x^\ast,\\
        1,  & \text{ if } x > x^\ast.\\
    \end{cases}
\end{align*} 
We adopt~\ref{item:A:h:2} to avoid the difficulty of choosing $h_2$ since relaxing this assumption may require a careful adjustment of the tail behavior of $h_2$ against $v_n$. 
The choice of $x^\ast$ will be discussed in Section~\ref{sec:numerical}.
\end{Remark}

Case~\ref{item:thm:case:II} in Theorem~\ref{thm:clt:hat:xi} requires the assumption $\kappa_1=\kappa_2$ for the asymptotic normality~\eqref{eq:clt:case:II} to hold.
Therefore, Case~\ref{item:thm:case:I} may be more desirable to hold for the purpose of measuring the differences in tail orders and in tail order parameters simultaneously.
In Appendix~\ref{sec:example}, we show that the \emph{Hill estimator}~\citep{hill1975simple} satisfies~\ref{item:A:tail:index:consistency} in Theorem~\ref{thm:consistency:hat:xi} and~\ref{item:mn:un:lim} and~\ref{item:A:tail:index:clt} in Theorem~\ref{thm:clt:hat:xi}.
Moreover, under a fairly general situation, there exists $a\ge 0$ and $b,b_1,b_2,\sigma_1,\sigma_2>0$ such that $m_n=O(n^\epsilon)$ satisfies all the assumptions in Theorems~\ref{thm:consistency:hat:xi} and~\ref{thm:clt:hat:xi} for any $\epsilon\in(a,b)$, and that
$\sqrt{k_l}(\hat \kappa_l-\kappa_l)\darrow \mathcal N(0,\sigma_l^2)$, $l=1,2$, where $\hat \kappa_l=\hat \kappa_l(k_l)$ is the Hill estimators of $\kappa_l$ and
$k_l=O(n^{\epsilon_l})$ for any $\epsilon_l\in(0,b_l)$.
In this situation, one can always take $\epsilon, \epsilon_1,\epsilon_2$ such that $\tilde m_n=k_1=k_2$, $m_n/\tilde m_n\rightarrow \infty$ and~\ref{item:A:tail:index:clt} holds with $\tilde \sigma^2=\{h_1'(\kappa_1-\kappa_2)\}^2(\sigma_1^2+\sigma_2^2)$, and thus the asymptotic normality~\eqref{eq:clt:case:I} in Case~\ref{item:thm:case:I} holds.
In addition, the asymptotic variance $\tilde\sigma^2$ in~\eqref{eq:clt:case:I} can be consistently estimated, based on which we can construct  confidence intervals for $\xi_{C_1,C_2,w}$.

\section{Numerical experiments}\label{sec:numerical}

In this section, we demonstrate numerical performance of our proposed measures and related hypothesis tests.
The hyperparameter $w$ is chosen to be $1/2$ throughout the experiment.
The auxiliary functions $h_1$ and $h_2$ are chosen to be the ones in Remarks~\ref{rem:range:xi},~\ref{rem:h2}, respectively, with $x^\ast=1.5$ in the simulation study and $x^\ast=0.5$ in the empirical study.
The choice of the threshold will also be discussed in the experiments.

\subsection{Simulation study}\label{sec:simulation}

We first conduct a simulation study for the proposed measures $\xi_{C_1,C_2,w}(u)$ and $\xi_{C_1,C_2,w}$ with copulas $C_1$ and $C_2$ belonging to some parametric families of copulas.
First, the skew normal copula~\citep{AzzaliniDallaValle1996,azzalini1999statistical,FungSeneta2016} is defined as the copula implicit in the multivariate skew normal distribution. 
Second, the AC skew $t$ copula~\citep{AzzaliniCapitanio2003,FungSeneta2010,Yoshiba2018} is also defined as the copula implicit in the multivariate AC skew $t$ distribution proposed by \cite{AzzaliniCapitanio2003}. The random vector $(Y_1, Y_2)$ of the bivariate skew normal distribution with parameters $\delta_1, \delta_2, \rho\in (-1,1)$ is constructed as:
\begin{equation*}
\begin{split}
Y_1 &= \delta_1 |Z_0| + \sqrt{1-\delta_1^2} Z_1,\\
Y_2 &= \delta_2 |Z_0| + \sqrt{1-\delta_2^2} Z_2,
\end{split}
\end{equation*}
where $(Z_0, Z_1, Z_2)$ is a trivariate normal random vector with standard normal margins, $Z_0$ being uncorrelated with $Z_1$ and $Z_2$, and $Z_1$ and $Z_2$ being correlated such that $\mathrm{corr}(Y_1,Y_2)=\rho$. The random vector $(X_1, X_2)$ of the bivariate AC skew $t$ distribution with parameters $\delta_1, \delta_2, \rho\in (-1,1)$ and $\nu>0$ is constructed as:
\begin{equation*}
X_1 = \frac{Y_1}{\sqrt{V}},\quad X_2 = \frac{Y_2}{\sqrt{V}},
\end{equation*}
using the Gamma random variable $V\sim \mathcal{G}(\nu/2, \nu/2)$ which is independent of the bivariate skew normal random vector $(Y_1, Y_2)$.
Using these copulas, we consider the following three cases which are expected to behave differently.

\begin{enumerate}[label=(\roman*)]\setcounter{enumi}{0}
 \item\label{item:case:1:indep:sn} This case compares lower tails of two identical skew normal copulas with parameters $\rho=0.5$ and $\delta_1=\delta_2=-0.4$.
 \item\label{item:case:2:sn} We compare lower and upper tails of skew normal copulas with $\rho=0.5$ and $\delta_1=\delta_2=-0.7$.
 \item\label{item:case:3:st}
 The last case compares lower and upper tails of AC skew $t$ copulas with $\rho=0.5$, $\nu=5$ and $\delta_1=\delta_2=0.6$.
 \end{enumerate}
True values of $\xi_{C_1,C_2,w}(u)$ and $\xi_{C_1,C_2,w}$ are available for all these cases; see~\cite{azzalini1999statistical,FungSeneta2016,FungSeneta2022,yoshiba2023measure} for skew normal copulas, and~\cite{FungSeneta2010,padoan2011multivariate} for AC skew $t$ copulas.
Note that, for the skew normal copulas, the tail order parameters do not take non-zero finite values.
Nevertheless, we can check by direct calculations that $\xi_{C_1,C_2,w}=0$ for Case~\ref{item:case:1:indep:sn} and $\xi_{C_1,C_2,w}\in (-1,-1/2)$ for Case~\ref{item:case:2:sn}.
We also have that $\xi_{C_1,C_2,w}\in (0,1/2)$  for Case~\ref{item:case:3:st}.
For all the cases, the sample size is $n=4\times 10^4$.

For the case of finite threshold, we compute an estimate of $\xi_{C_1,C_2,w}(u)$ and its asymptotic variance for $u\in[0.0025,0.25]$ based on~Theorem~\ref{thm:asymptotics}.
For Case~\ref{item:case:1:indep:sn}, samples from $C_1$ and $C_2$ are generated to be mutually independent, that is, the underlying copula of $H$ is the independence copula.
For other cases where $C_1=C$ is either the skew normal or skew $t$ copula, we first draw iid samples $\mathbf{U}_i^{(1)}$, $i=1,\dots,n$, from $C$, and samples from $\hat C$ are generated by  $\mathbf{U}_i^{(2)}=\mathbf{1}_2-\mathbf{U}_i^{(1)}$, $i=1,\dots,n$.
As a result, the copula of $H$ is the counter-monotonic copula.
The results are summarized in Figure~\ref{fig:sim:fixed}.
The second column in this figure summarizes the asymptotic p-values of the null-hypothesis that $\xi_{C_1,C_2,w}=0$.
Three different p-values are plotted for two one-tailed tests and the two-tailed test.
Namely, the black solid lines represent $\Phi(T_n(u))$, where $T_n(u)=\sqrt{n}\hat \xi_w(u)/\sqrt{\hat \sigma_{w}(u)}$, the red solid lines show $1-\Phi(T_n(u))$ and the blue solid lines indicate $2\,\min\left(\Phi(T_n(u)),1-\Phi(T_n(u))\right)$.

\begin{figure}[!t]
\begin{center}
\includegraphics[width=123mm]{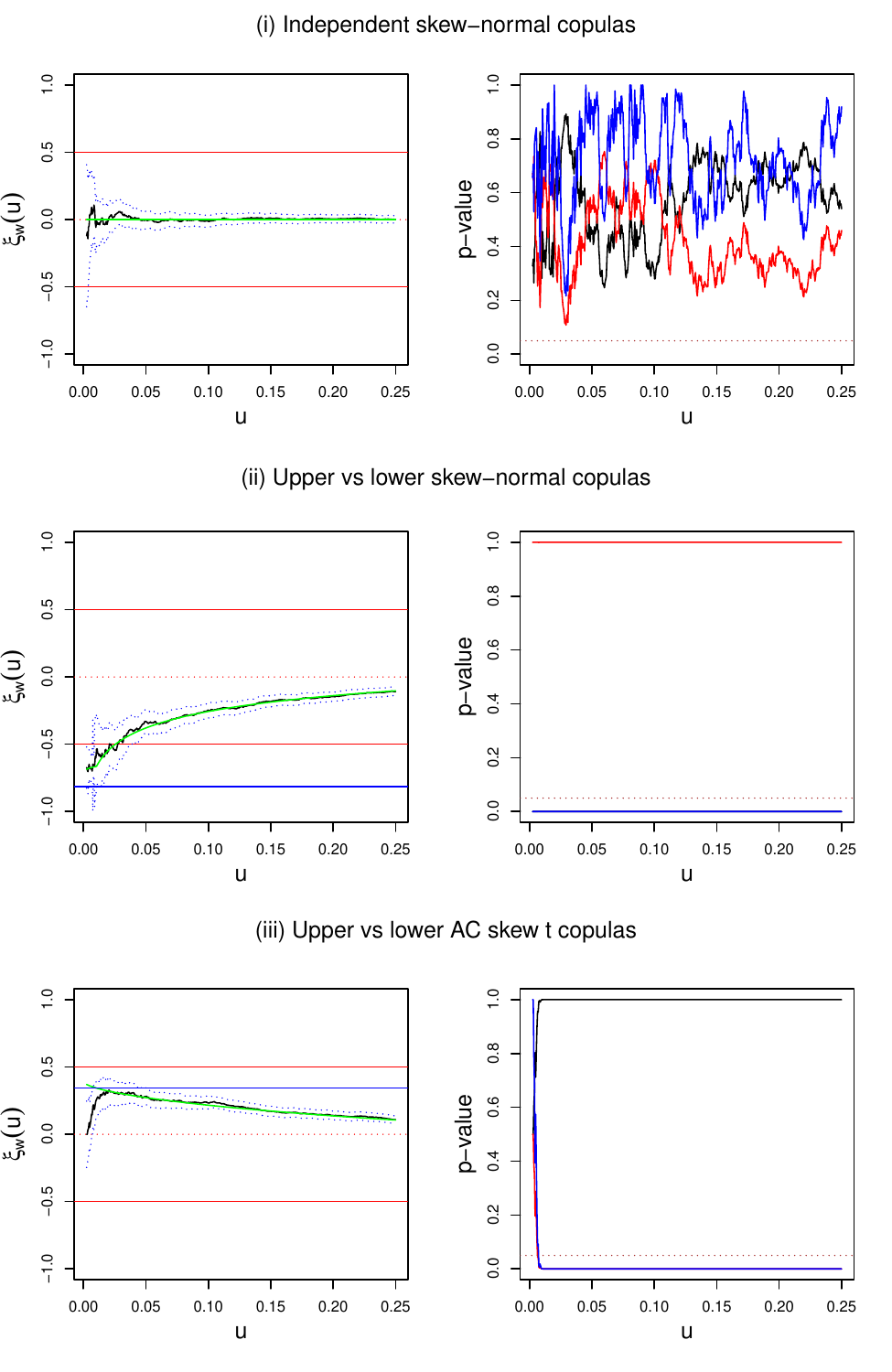}
\caption{(First column) Estimates of $\xi_{C_1,C_2,w}(u)$ (solid black lines) for $u$ from $0.0025$ to $0.25$ together with their 95\% confidence intervals (dotted blue lines). Solid and dotted red lines represent $y=-0.5$, $0$ and $0.5$.
The solid green and blue lines show the true values of $\xi_{C_1,C_2,w}(u)$ and $\xi_{C_1,C_2,w}$, respectively.
(Second column) Asymptotic p-values of the left-tailed test (black), right-tailed test (red) and the two-tailed test (blue). The brown dotted lines represent the standard critical level $0.05$.}
\label{fig:sim:fixed}
\end{center}
\end{figure} 

For the case of varying thresholds, we compute an estimate of $\xi_{C_1,C_2,w}$  and its asymptotic variance based on~Theorem~\ref{thm:clt:hat:xi}.
For all the cases, samples from $C_1$ and $C_2$ are drawn to be mutually independent.
For estimating the tail orders, we adopt the Hill estimator, whose asymptotic variance is consistently estimated as in Theorem~\ref{thm:hill} in~Appendix~\ref{sec:example}.
We vary $\tilde m_n$ and $v_n$ such that $(\tilde m_n,v_n)\in\{(1000,v),(k,0.025), v\in[0.01,0.1], k\in [400,4000]\}$, and take the thresholds $k_l$, $l=1,2,$ for Hill estimators of $\kappa_l$ such that $k_1=k_2=\tilde m_n$.
Furthermore, we assume that $m_n/\tilde m_n\rightarrow\infty$ and estimate the asymptotic variance of~\eqref{eq:clt:case:I} in Case~\ref{item:thm:case:I} by~Theorem~\ref{thm:hill}, based on which p-values are calculated analogously to the case of the finite threshold.
The results are summarized in Figures~\ref{fig:sim:vary:k} and~\ref{fig:sim:vary:v}.

\begin{figure}[!t]
\begin{center}
\includegraphics[width=127mm]{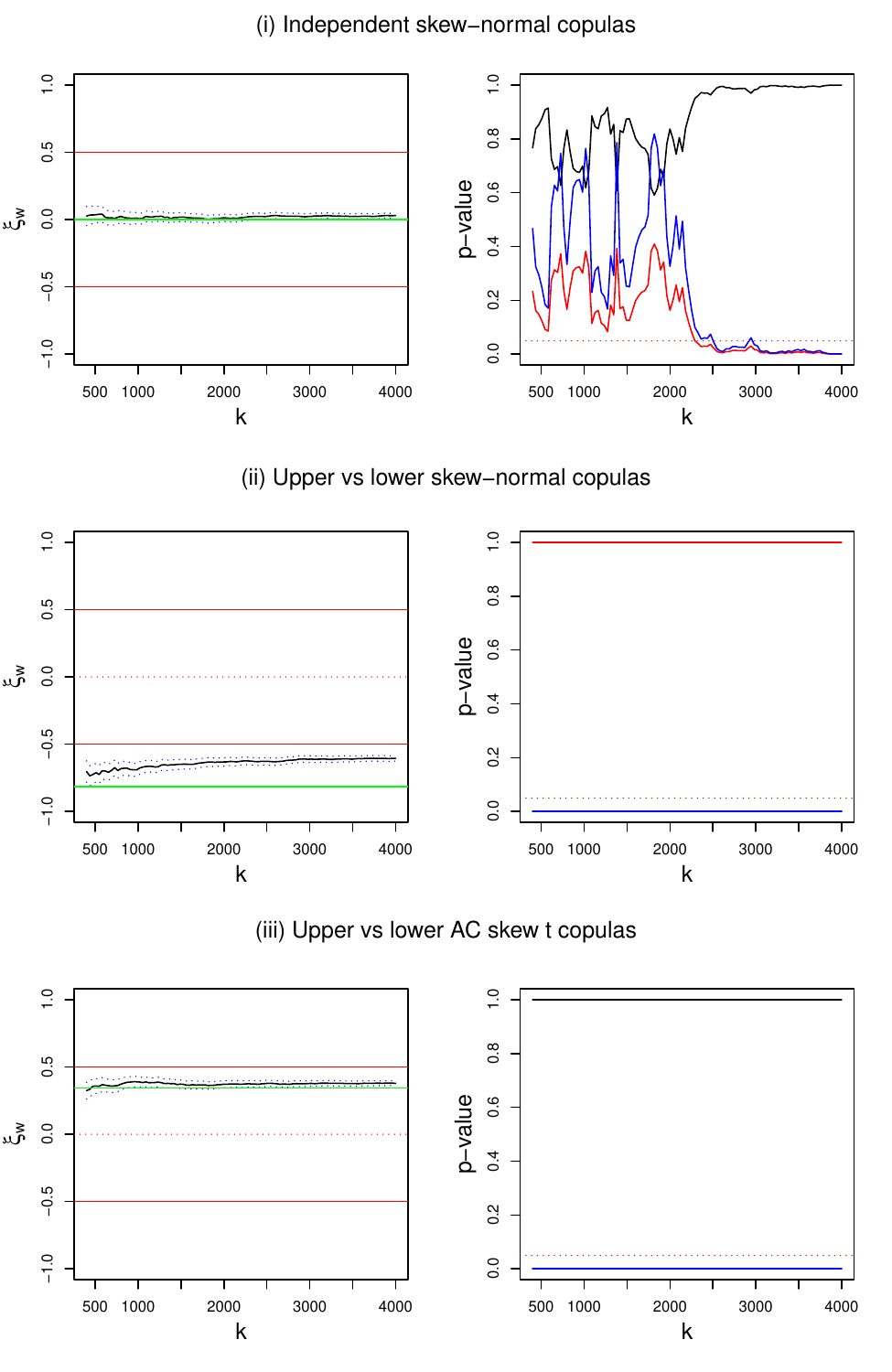}
\caption{
(First column) Estimates of $\xi_{C_1,C_2,w}$ (solid black lines) for $v_n=0.025$ and $\tilde m_n=k$ varying from $400$ to $4000$ together with their 95\% confidence intervals (dotted blue lines). Solid and dotted red lines represent $y=-0.5$, $0$ and $0.5$.
The solid green lines show the true values of $\xi_{C_1,C_2,w}$. 
(Second column) Asymptotic p-values of the left-tailed test (black), right-tailed test (red) and the two-tailed test (blue). The brown dotted line represents the standard critical level $0.05$.
}
\label{fig:sim:vary:k}
\end{center}
\end{figure}  

\begin{figure}[!t]
\begin{center}
\includegraphics[width=127mm]{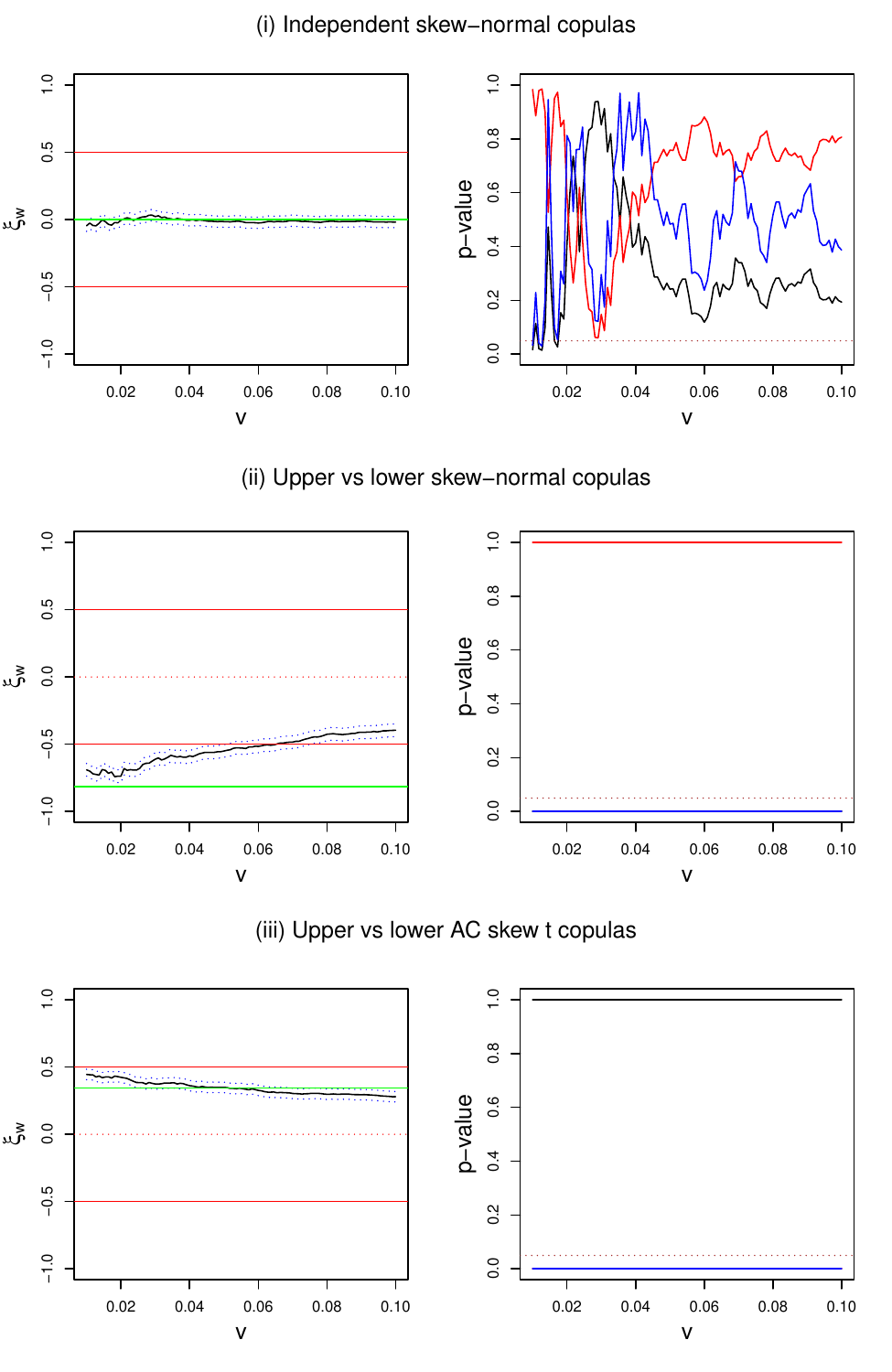}
\caption{
(First column) Estimates of $\xi_{C_1,C_2,w}$ (solid black lines) for $\tilde m_n=k=1000$ and $v$ varying from $0.01$ to $0.1$ together with their 95\% confidence intervals (dotted blue lines). Solid and dotted red lines represent $y=-0.5$, $0$ and $0.5$.
The solid green lines show the true values of $\xi_{C_1,C_2,w}$. 
(Second column) Asymptotic p-values of the left-tailed test (black), right-tailed test (red) and the two-tailed test (blue).
The brown dotted line represents the standard critical level $0.05$.
}
\label{fig:sim:vary:v}
\end{center}
\end{figure}  

In Figure~\ref{fig:sim:fixed}, we observe that estimates of $\hat \xi_{C_1,C_2,w}(u)$ are close to the true values for all the cases.
We also see that a suitable threshold $u$ differs in cases for $\xi_{C_1,C_2,w}(u)$ to approximate the limit $\xi_{C_1,C_2,w}$.
In Case~\ref{item:case:1:indep:sn}, the confidence interval gets wider as $u$ gets smaller.
In Case~\ref{item:case:2:sn}, the convergence of $\xi_{C_1,C_2,w}(u)$ to $\xi_{C_1,C_2,w}$ seems slower when compared with Case~\ref{item:case:3:st}.
As a result, smaller $u$ seems required in Case~\ref{item:case:2:sn}, and too small $u$ may fail to estimate $\xi_{C_1,C_2,w}$ in Case~\ref{item:case:3:st}.
Comparing Figures~\ref{fig:sim:vary:k} and~\ref{fig:sim:vary:v}, the estimates of $\xi_{C_1,C_2,w}$ are more stable in Figure~\ref{fig:sim:vary:k}, and thus the choice of $v_n$ seems more crucial than that of $k$.
In both figures, we find that estimates in Case~\ref{item:case:2:sn} are biased, which may be as expected from the slow convergence. 
Regarding the p-values, we observe for Case~\ref{item:case:1:indep:sn} that the null-hypothesis is rarely rejected, and this result is consistent with the theoretical value $\xi_{C_1,C_2,w}=0$.
For Cases~\ref{item:case:2:sn} and~\ref{item:case:3:st}, the null-hypothesis is rejected with the standard critical level $0.05$ based on all the calculated p-values of the two-sided test and the one-sided test with the side correctly specified.

For robustness analysis, we compute estimates of $\xi_{C_1,C_2,w}$ with $k=1000$ and $v\in[0.01,0.1]$ for $x^\ast$ varying from $1$ to $2$.
The results are given in Figure~\ref{fig:sim:vary:xast}.
We see that, in Case~\ref{item:case:1:indep:sn}, the estimates are stable against the choice of $x^\ast$.
In Case~\ref{item:case:2:sn}, the curves $v\mapsto \hat \xi_{w}$ move away from $0$ as $v$ gets smaller, and become stable for sufficiently small $v$.
In Case~\ref{item:case:3:st}, the curves seem to be gradually converging as $v$ gets smaller.
These features can be generally expected depending on the case of identical or non-identical tail orders.
It is recommended to change $x^\ast$ when the curve $v\mapsto \hat \xi_{w}$ fluctuates around one of the boundaries $\pm w$ since, in such a case, it is not clear whether the difference of tail behaviors stems from that of tail orders or of tail order parameters.

\begin{figure}[!t]
\begin{center}
\includegraphics[width=165mm]{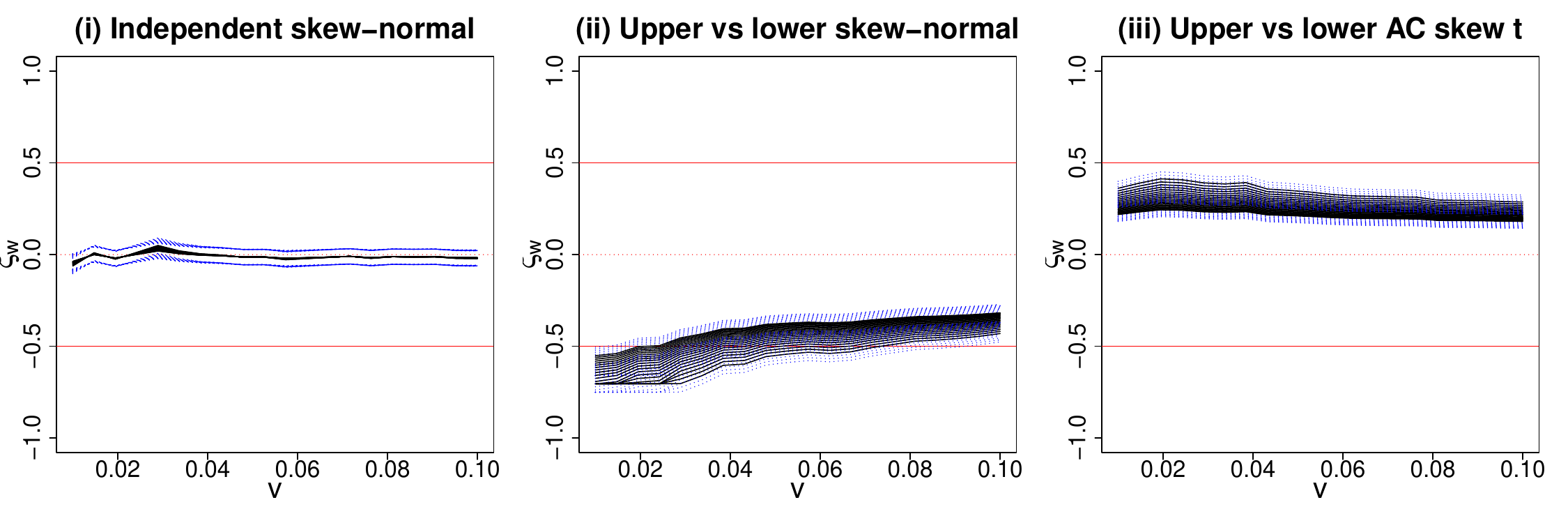}
\caption{Estimates of $\xi_{C_1,C_2,w}$ (solid black lines) and their 95\% confidence intervals (dotted blue lines) for $x^\ast$ varying from $1$ to $2$.
We fix $k=1000$, and $v$ varies from $0.01$ to $0.1$.
Solid and dotted red lines represent $y=-0.5$, $0$ and $0.5$.
}
\label{fig:sim:vary:xast}
\end{center}
\end{figure}

\subsection{Empirical study}\label{sec:empirical}

In this section, we apply the inference in Section~\ref{sec:simulation} to a real data.
We analyze negative daily log returns of the stock indices SP500, FTSE and NIKKEI where we consider two periods, one from January 1st, 2006 to December 9th, 2010, and the other from January 1st, 2018 to December 31st, 2022.
The first period contains the global financial crisis in 2007, and the second includes the COVID-19 recession.
The sample size in each period is $n=1153$.
In each period, we fit GARCH(1,1) model with skew $t$ white noise to filter the residuals, which are assumed to be independent and identically distributed over time.
We then transform the residuals to pseudo-observations, based on which we quantify and test tail asymmetry with our proposed measures and related hypothesis tests.
We consider the following five different pairs of tails.
\begin{enumerate}[label=(\roman*)]\setcounter{enumi}{0}
 \item\label{item:case:i:emp} Lower and upper tails of (FTSE, NIKKEI) in the first period.
\item\label{item:case:ii:emp} Lower and upper tails of (FTSE, NIKKEI) in the second period.
\item\label{item:case:iii:emp} Upper tails of (SP500, 
 NIKKEI) and (SP500, FTSE) in the first period.
\item\label{item:case:iv:emp} Upper tails of (SP500, 
 NIKKEI) and (SP500, FTSE) in the second period.
\item\label{item:case:v:emp} Upper tail of (FTSE, NIKKEI) in the first and second periods. 
\end{enumerate}

When estimating the limit $\xi_{C_1,C_2,w}$, we split the samples to reduce dependence between samples from $C_1$ and those from $C_2$ in Cases~\ref{item:case:i:emp}--\ref{item:case:iv:emp}.
We do not split samples when comparing upper tails in Case~\ref{item:case:v:emp} under the independence assumption between samples in the first and second periods.
Such a splitting is also unnecessary when estimating $\xi_{C_1,C_2,w}(u)$.
The hyperparameters are set to be the same as those in the simulation study.
Moreover, we vary $\tilde m_n$ and $v_n$ such that $(\tilde m_n,v_n)\in\{(100,v),(k,0.1), v\in[0.05,0.4], k\in [50,500]\}$.
The results are summarized in Figures~\ref{fig:emp:fixed},~\ref{fig:emp:vary:k} and~\ref{fig:emp:vary:v}.
For robustness analysis, we also compute $\hat \xi_{w}$ with $k=100$ and $v \in [0.1,0.2]$ for $x^\ast\in [0.25,1.5]$.
The results are illustrated in Figure~\ref{fig:emp:vary:xast}.

\begin{figure}[!t]
\begin{center}
\includegraphics[width=105mm]{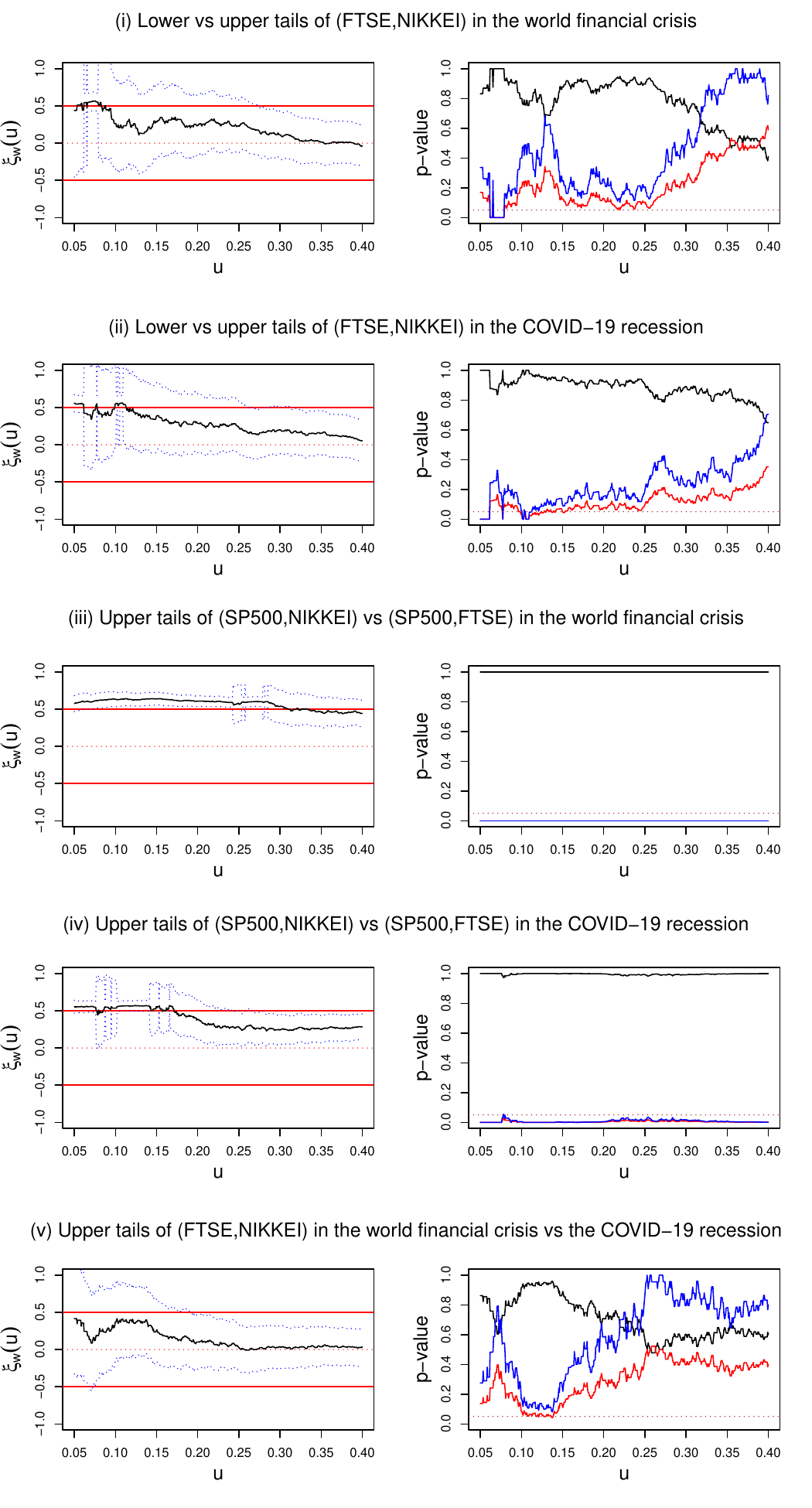}
\caption{(First column) Estimates of $\xi_{C_1,C_2,w}(u)$ (solid black lines), $u\in[0.05,0.4]$, in various tails and time periods together with their 95\% confidence intervals (dotted blue lines).
Solid and dotted red lines represent $y=-0.5$, $0$ and $0.5$.
(Second column) Asymptotic p-values of the left-tailed test (black), right-tailed test (red) and the two-tailed test (blue). The brown dotted line represents the standard critical level $0.05$.}
\label{fig:emp:fixed}
\end{center}
\end{figure} 

\begin{figure}[!t]
\begin{center}
\includegraphics[width=105mm]{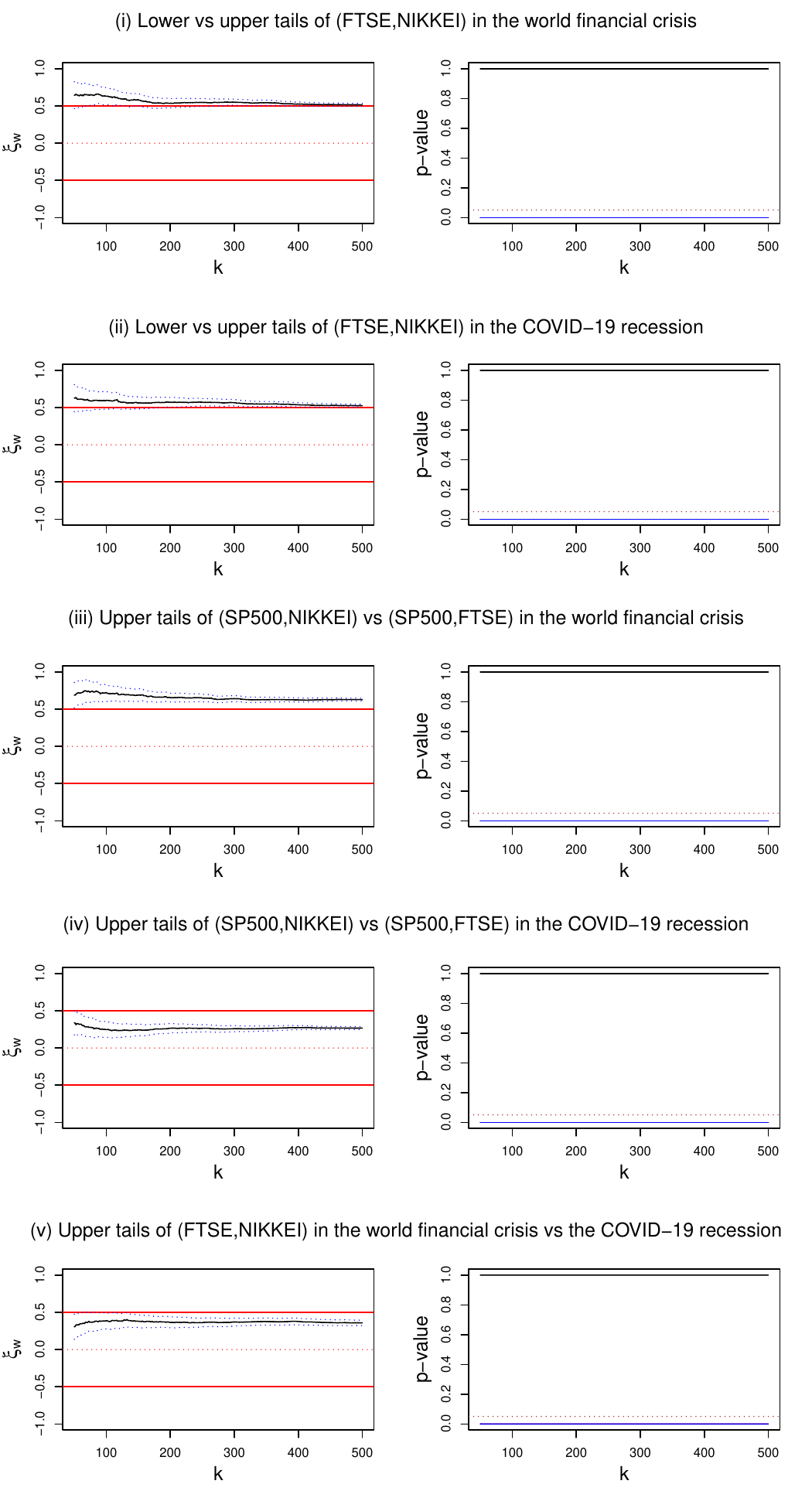}
\caption{(First column) Estimated $\xi_{C_1,C_2,w}$ (solid black lines) in various tails and time periods together with their 95\% confidence intervals (dotted blue lines).
Solid and dotted red lines represent $y=-0.5$, $0$ and $0.5$.
We fix $v=0.1$, and $k$ varies from $50$ to $500$.
(Second column) Asymptotic p-values of the left-tailed test (black), right-tailed test (red) and the two-tailed test (blue). The brown dotted line represents the standard critical level $0.05$.}
\label{fig:emp:vary:k}
\end{center}
\end{figure}  

\begin{figure}[!t]
\begin{center}
\includegraphics[width=105mm]{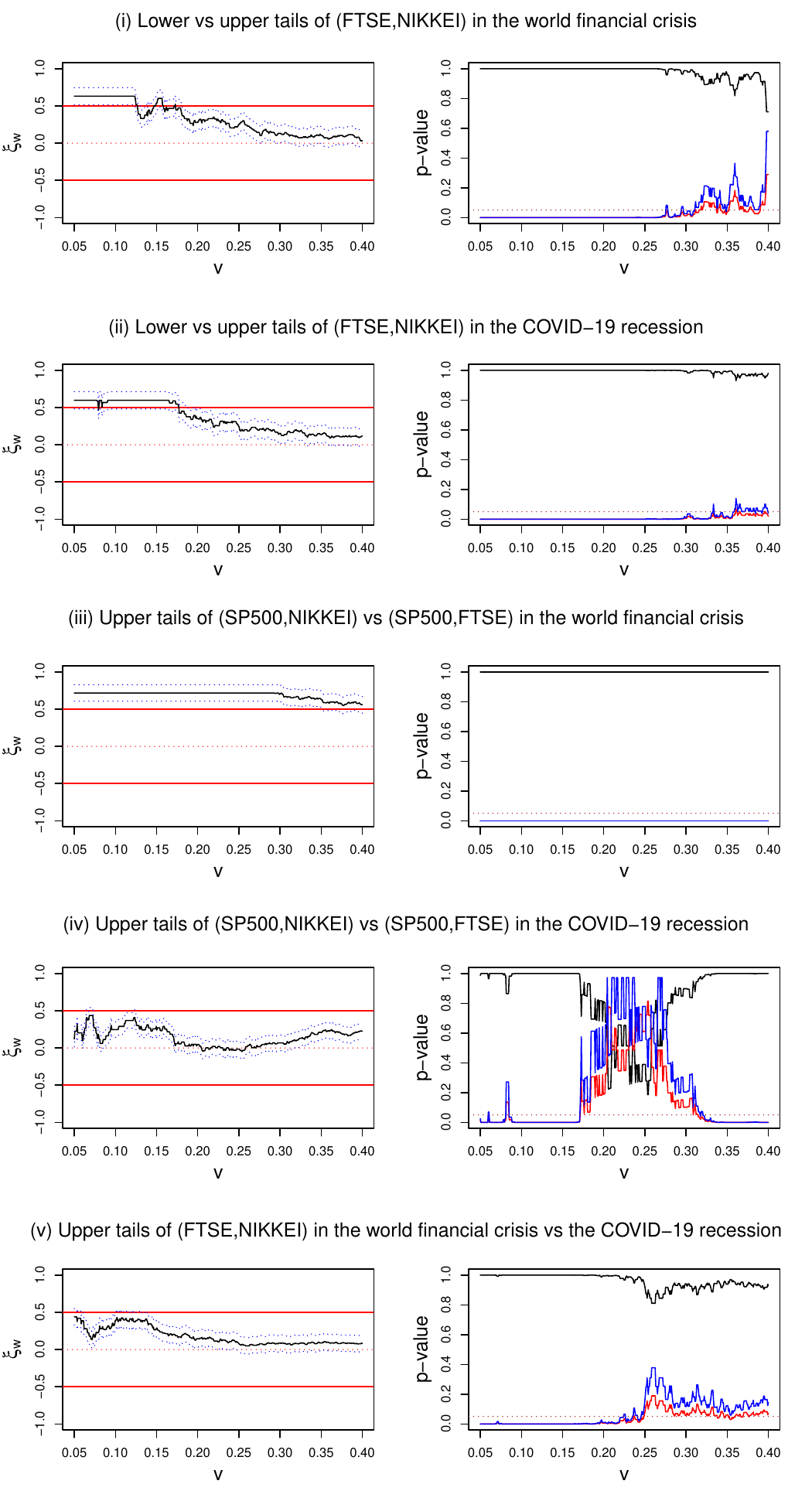}
\caption{(First column) Estimated $\xi_{C_1,C_2,w}$ (solid black lines) in various tails and time periods together with their 95\% confidence intervals (dotted blue lines).
Solid and dotted red lines represent $y=-0.5$, $0$ and $0.5$.
We fix $k=100$, and $v$ varies from $0.05$ to $0.4$.
(Second column) Asymptotic p-values of the left-tailed test (black), right-tailed test (red) and the two-tailed test (blue). The brown dotted line represents the standard critical level $0.05$.}
\label{fig:emp:vary:v}
\end{center}
\end{figure}

\begin{figure}[!t]
\begin{center}
\includegraphics[width=165 mm]{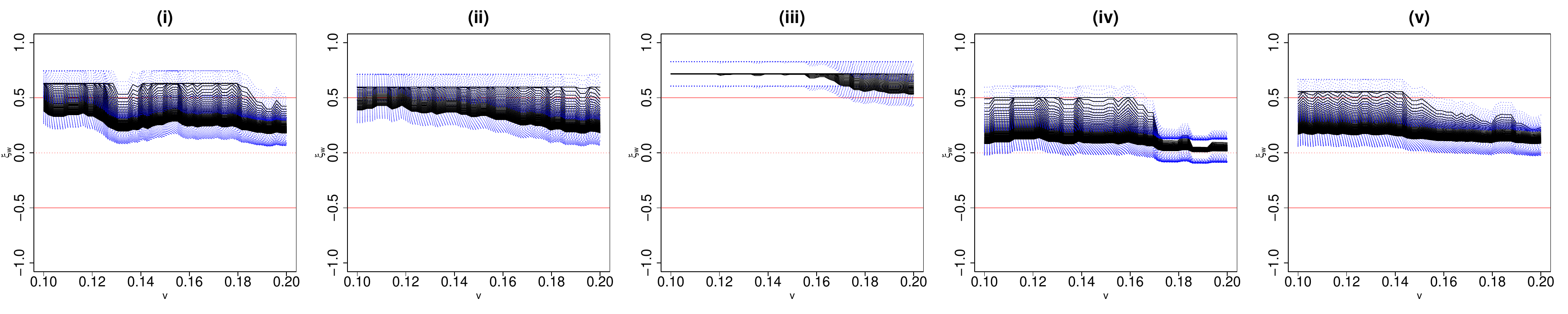}
\caption{Estimates of $\xi_{C_1,C_2,w}$ (solid black lines) and their 95\% confidence intervals with $k=100$ and $v\in[0.1,0.2]$ for $x^\ast$ varying from $0.25$ to $1.5$.
Solid and dotted red lines represent $y=-0.5$, $0$ and $0.5$.
}
\label{fig:emp:vary:xast}
\end{center}
\end{figure} 

We first focus on Cases~\ref{item:case:i:emp} and~\ref{item:case:ii:emp}, where upper and lower tails of (FTSE,NIKKEI) is considered in different periods.
For both cases, we see in Figure~\ref{fig:emp:fixed} that $\hat \xi_w(u)$ ascends to $1/2$ as $u$ gets smaller and fluctuates around $1/2$ for small thresholds.
Therefore, for both cases, the upper tail dependence is stronger than the lower one.
Note that the confidence intervals can be discontinuous with respect to $u$ since $h'_2$ jumps at $\pm x^\ast$.
Next, Figures~\ref{fig:emp:vary:k} and~\ref{fig:emp:vary:v} show that the estimates of the limit $\xi_{C_1,C_2,w}$ are stable slightly above $1/2$.
Together with the robustness analysis in Figure~\ref{fig:emp:vary:xast}, we draw a vague conclusion that both cases exhibit either moderately different tail orders or the identical tail order with  different tail order parameters.

We next compare Cases~\ref{item:case:iii:emp} and~\ref{item:case:iv:emp}, where we consider the two upper tails of different pairs of stock indices in each period.
For these cases, the tail dependence of (SP500,FTSE) is stronger than that of (SP500,NIKKEI) in both periods of time.
In Cases~\ref{item:case:iii:emp}, all the estimates for small thresholds are clearly above $1/2$, which indicates different tail orders between (SP500,FTSE) and (SP500,NIKKEI).
On the other hands, the estimates fluctuates around $1/2$ in Case~\ref{item:case:iv:emp}.
Combined with the robustness analysis in Figure~\ref{fig:emp:vary:xast}, we find that Case~\ref{item:case:iii:emp} exhibits different tail orders and Case~\ref{item:case:iv:emp} has the identical tail order with different tail order parameters.

In Case~\ref{item:case:v:emp}, estimates of $\xi_{C_1,C_2,w}(u)$ are close to $0$ for moderate thresholds and are slightly above $0$ for small thresholds, which indicates weak asymmetry in the upper tails between two different periods of recession.
This weak asymmetry in the tails is supported more clearly by the analyses of varying thresholds and $x^\ast$.
Since the estimates for small thresholds are below $1/2$, we conclude that, in this case, the tail order remains unchanged and the tail order parameters differ in two times of recession.

\section{Concluding remarks}\label{sec:conclusion}

We propose new measures signifying the qualitative difference of two tail behaviors summarized by tail orders and tail order parameters. 
Summarizing such a difference by a single number can be particularly beneficial when, for example, dynamic alternation of such a difference is of interest.
Statistical inference of the proposed measures can be conducted based on our asymptotic results.

As is commonly the case for extreme value analysis, the choice of thresholds is an important issue for the statistical inference on $\xi_{C_1,C_2,w}$.
From the numerical illustrations in Section~\ref{sec:numerical}, our suggestion is to first search a suitable $u\in(0,1)$ by plotting $\hat \xi_{w}(u)$ against $u$, and then check the stability of $\hat \xi_w$ for various pairs of thresholds $(k,v)$.
Particularly when our measures take values around $1/2$ or $-1/2$, it is recommendable to change $x^\ast$ in $h_2$ for the measures to detect the qualitative difference of tail behaviors more clearly.

\section*{Acknowledgements}
This research was funded by JSPS KAKENHI Grant Number
JP21K13275 (Koike),
JP20K03759 (Kato), 
JP21K01581 (Yoshiba) and
JP24K00273 (Koike, Kato and Yoshiba).

\bibliographystyle{elsarticle-num-names}

\begin{thebibliography}{28}
\expandafter\ifx\csname natexlab\endcsname\relax\def\natexlab#1{#1}\fi
\providecommand{\url}[1]{\texttt{#1}}
\providecommand{\href}[2]{#2}
\providecommand{\path}[1]{#1}
\providecommand{\DOIprefix}{doi:}
\providecommand{\ArXivprefix}{arXiv:}
\providecommand{\URLprefix}{URL: }
\providecommand{\Pubmedprefix}{pmid:}
\providecommand{\doi}[1]{\href{http://dx.doi.org/#1}{\path{#1}}}
\providecommand{\Pubmed}[1]{\href{pmid:#1}{\path{#1}}}
\providecommand{\bibinfo}[2]{#2}
\ifx\xfnm\relax \def\xfnm[#1]{\unskip,\space#1}\fi
\bibitem[{Ang and Chen(2002)}]{AngChen2002}
\bibinfo{author}{A.~Ang}, \bibinfo{author}{J.~Chen},
\newblock \bibinfo{title}{Asymmetric correlations of equity portfolios},
\newblock \bibinfo{journal}{Journal of Financial Economics}
  \bibinfo{volume}{63} (\bibinfo{year}{2002}) \bibinfo{pages}{443--494}.
\bibitem[{Jondeau(2016)}]{jondeau2016asymmetry}
\bibinfo{author}{E.~Jondeau},
\newblock \bibinfo{title}{Asymmetry in tail dependence in equity portfolios},
\newblock \bibinfo{journal}{Computational Statistics \& Data Analysis}
  \bibinfo{volume}{100} (\bibinfo{year}{2016}) \bibinfo{pages}{351--368}.
\bibitem[{Bormann and Schienle(2020)}]{bormann2020detecting}
\bibinfo{author}{C.~Bormann}, \bibinfo{author}{M.~Schienle},
\newblock \bibinfo{title}{Detecting structural differences in tail dependence
  of financial time series},
\newblock \bibinfo{journal}{Journal of Business \& Economic Statistics}
  \bibinfo{volume}{38} (\bibinfo{year}{2020}) \bibinfo{pages}{380--392}.
\bibitem[{Can et~al.(2024)Can, Einmahl, and Laeven}]{can2024two}
\bibinfo{author}{S.~U. Can}, \bibinfo{author}{J.~H. Einmahl},
  \bibinfo{author}{R.~J. Laeven},
\newblock \bibinfo{title}{Two-sample testing for tail copulas with an
  application to equity indices},
\newblock \bibinfo{journal}{Journal of Business \& Economic Statistics}
  \bibinfo{volume}{42} (\bibinfo{year}{2024}) \bibinfo{pages}{147--159}.
\bibitem[{Nelsen(2006)}]{nelsen2006introduction}
\bibinfo{author}{R.~B. Nelsen}, \bibinfo{title}{An Introduction to Copulas},
  \bibinfo{publisher}{Springer}, \bibinfo{address}{New York},
  \bibinfo{year}{2006}.
\bibitem[{Sibuya(1960)}]{sibuya1960bivariate}
\bibinfo{author}{M.~Sibuya},
\newblock \bibinfo{title}{Bivariate extreme statistics, i},
\newblock \bibinfo{journal}{Annals of the Institute of Statistical Mathematics}
  \bibinfo{volume}{11} (\bibinfo{year}{1960}) \bibinfo{pages}{195--210}.
\bibitem[{Hua and Joe(2011)}]{HuaJoe2011}
\bibinfo{author}{L.~Hua}, \bibinfo{author}{H.~Joe},
\newblock \bibinfo{title}{Tail order and intermediate tail dependence of
  multivariate copulas},
\newblock \bibinfo{journal}{Journal of Multivariate Analysis}
  \bibinfo{volume}{102} (\bibinfo{year}{2011}) \bibinfo{pages}{1454--1471}.
\bibitem[{Adcock and Azzalini(2020)}]{adcock2020selective}
\bibinfo{author}{C.~Adcock}, \bibinfo{author}{A.~Azzalini},
\newblock \bibinfo{title}{A selective overview of skew-elliptical and related
  distributions and of their applications},
\newblock \bibinfo{journal}{Symmetry} \bibinfo{volume}{12}
  (\bibinfo{year}{2020}) \bibinfo{pages}{118}.
\bibitem[{Azzalini(2022)}]{azzalini2022overview}
\bibinfo{author}{A.~Azzalini},
\newblock \bibinfo{title}{An overview on the progeny of the skew-normal
  family—a personal perspective},
\newblock \bibinfo{journal}{Journal of Multivariate Analysis}
  \bibinfo{volume}{188} (\bibinfo{year}{2022}) \bibinfo{pages}{104851}.
\bibitem[{Kato et~al.(2022)Kato, Yoshiba, and Eguchi}]{Kato_etal2022}
\bibinfo{author}{S.~Kato}, \bibinfo{author}{T.~Yoshiba},
  \bibinfo{author}{S.~Eguchi},
\newblock \bibinfo{title}{Copula-based measures of asymmetry between the lower
  and upper tail probabilities},
\newblock \bibinfo{journal}{Statistical Papers} \bibinfo{volume}{63}
  (\bibinfo{year}{2022}) \bibinfo{pages}{1907--1929}.
\bibitem[{Falk et~al.(2010)Falk, H{\"u}sler, and Reiss}]{falk2010laws}
\bibinfo{author}{M.~Falk}, \bibinfo{author}{J.~H{\"u}sler},
  \bibinfo{author}{R.-D. Reiss}, \bibinfo{title}{Laws of small numbers:
  extremes and rare events}, \bibinfo{publisher}{Springer Science \& Business
  Media}, \bibinfo{year}{2010}.
\bibitem[{Embrechts et~al.(2013)Embrechts, Kl{\"u}ppelberg, and
  Mikosch}]{embrechts2013modelling}
\bibinfo{author}{P.~Embrechts}, \bibinfo{author}{C.~Kl{\"u}ppelberg},
  \bibinfo{author}{T.~Mikosch}, \bibinfo{title}{Modelling Extremal Events: for
  Insurance and Finance}, volume~\bibinfo{volume}{33},
  \bibinfo{publisher}{Springer Science \& Business Media},
  \bibinfo{year}{2013}.
\bibitem[{Schmid and Schmidt(2007)}]{schmid2007nonparametric}
\bibinfo{author}{F.~Schmid}, \bibinfo{author}{R.~Schmidt},
\newblock \bibinfo{title}{Nonparametric inference on multivariate versions of
  blomqvist's beta and related measures of tail dependence},
\newblock \bibinfo{journal}{Metrika} \bibinfo{volume}{66}
  (\bibinfo{year}{2007}) \bibinfo{pages}{323--354}.
\bibitem[{Hill(1975)}]{hill1975simple}
\bibinfo{author}{B.~M. Hill},
\newblock \bibinfo{title}{A simple general approach to inference about the tail
  of a distribution},
\newblock \bibinfo{journal}{Annals of Statistics} \bibinfo{volume}{3}
  (\bibinfo{year}{1975}) \bibinfo{pages}{1163--1174}.
\bibitem[{Azzalini and {Dalla Valle}(1996)}]{AzzaliniDallaValle1996}
\bibinfo{author}{A.~Azzalini}, \bibinfo{author}{A.~{Dalla Valle}},
\newblock \bibinfo{title}{The multivariate skew-normal distribution},
\newblock \bibinfo{journal}{Biometrika} \bibinfo{volume}{83}
  (\bibinfo{year}{1996}) \bibinfo{pages}{715--726}.
\bibitem[{Azzalini and Capitanio(1999)}]{azzalini1999statistical}
\bibinfo{author}{A.~Azzalini}, \bibinfo{author}{A.~Capitanio},
\newblock \bibinfo{title}{Statistical applications of the multivariate skew
  normal distribution},
\newblock \bibinfo{journal}{Journal of the Royal Statistical Society: Series B
  (Statistical Methodology)} \bibinfo{volume}{61} (\bibinfo{year}{1999})
  \bibinfo{pages}{579--602}.
\bibitem[{Fung and Seneta(2016)}]{FungSeneta2016}
\bibinfo{author}{T.~Fung}, \bibinfo{author}{E.~Seneta},
\newblock \bibinfo{title}{Tail asymptotics for the bivariate skew normal},
\newblock \bibinfo{journal}{Journal of Multivariate Analysis}
  \bibinfo{volume}{144} (\bibinfo{year}{2016}) \bibinfo{pages}{129--138}.
\bibitem[{Azzalini and Capitanio(2003)}]{AzzaliniCapitanio2003}
\bibinfo{author}{A.~Azzalini}, \bibinfo{author}{A.~Capitanio},
\newblock \bibinfo{title}{Distributions generated by perturbation of symmetry
  with emphasis on a multivariate skew $t$-distribution},
\newblock \bibinfo{journal}{Journal of the Royal Statistical Society Series B}
  \bibinfo{volume}{65} (\bibinfo{year}{2003}) \bibinfo{pages}{367--389}.
\bibitem[{Fung and Seneta(2010)}]{FungSeneta2010}
\bibinfo{author}{T.~Fung}, \bibinfo{author}{E.~Seneta},
\newblock \bibinfo{title}{Tail dependence for two skew $t$ distributions},
\newblock \bibinfo{journal}{Statistics \& Probability Letters}
  \bibinfo{volume}{80} (\bibinfo{year}{2010}) \bibinfo{pages}{784--791}.
\bibitem[{Yoshiba(2018)}]{Yoshiba2018}
\bibinfo{author}{T.~Yoshiba},
\newblock \bibinfo{title}{Maximum likelihood estimation of skew-$t$ copulas
  with its applications to stock returns},
\newblock \bibinfo{journal}{Journal of Statistical Computation and Simulations}
  \bibinfo{volume}{88} (\bibinfo{year}{2018}) \bibinfo{pages}{2489--2506}.
\bibitem[{Fung and Seneta(2022)}]{FungSeneta2022}
\bibinfo{author}{T.~Fung}, \bibinfo{author}{E.~Seneta},
\newblock \bibinfo{title}{Tail asymptotics for the bivariate skew normal in the
  general case},
\newblock \bibinfo{journal}{arXiv preprint arXiv:2210.01284}
  (\bibinfo{year}{2022}).
\bibitem[{Yoshiba et~al.(2023)Yoshiba, Koike, and Kato}]{yoshiba2023measure}
\bibinfo{author}{T.~Yoshiba}, \bibinfo{author}{T.~Koike},
  \bibinfo{author}{S.~Kato},
\newblock \bibinfo{title}{On a measure of tail asymmetry for the bivariate
  skew-normal copula},
\newblock \bibinfo{journal}{Symmetry} \bibinfo{volume}{15}
  (\bibinfo{year}{2023}) \bibinfo{pages}{1410}.
\bibitem[{Padoan(2011)}]{padoan2011multivariate}
\bibinfo{author}{S.~A. Padoan},
\newblock \bibinfo{title}{Multivariate extreme models based on underlying
  skew-t and skew-normal distributions},
\newblock \bibinfo{journal}{Journal of Multivariate Analysis}
  \bibinfo{volume}{102} (\bibinfo{year}{2011}) \bibinfo{pages}{977--991}.
\bibitem[{Resnick(2007)}]{Resnick2007}
\bibinfo{author}{S.~I. Resnick}, \bibinfo{title}{Heavy-Tail Phenomena:
  Probabilistic and Statistical Modeling}, Springer Series in Operations
  Research and Financial Engineering, \bibinfo{publisher}{Springer},
  \bibinfo{address}{New York, NY}, \bibinfo{year}{2007}.
\bibitem[{van~der Vaart(1998)}]{van1998asy}
\bibinfo{author}{A.~W. van~der Vaart}, \bibinfo{title}{Asymptotic Statistics},
  \bibinfo{publisher}{Cambridge University Press}, \bibinfo{year}{1998}.
\bibitem[{Durrett(2019)}]{durrett2019probability}
\bibinfo{author}{R.~Durrett}, \bibinfo{title}{Probability: Theory and
  Examples}, \bibinfo{publisher}{Cambridge University Press},
  \bibinfo{year}{2019}.
\bibitem[{Billingsley(2017)}]{billingsley2017probability}
\bibinfo{author}{P.~Billingsley}, \bibinfo{title}{Probability and measure},
  \bibinfo{publisher}{John Wiley \& Sons}, \bibinfo{year}{2017}.
\bibitem[{Draisma et~al.(2004)Draisma, Drees, Ferreira, and
  De~Haan}]{draisma2004bivariate}
\bibinfo{author}{G.~Draisma}, \bibinfo{author}{H.~Drees},
  \bibinfo{author}{A.~Ferreira}, \bibinfo{author}{L.~De~Haan},
\newblock \bibinfo{title}{Bivariate tail estimation: dependence in asymptotic
  independence},
\newblock \bibinfo{journal}{Bernoulli} \bibinfo{volume}{10}
  (\bibinfo{year}{2004}) \bibinfo{pages}{251--280}.

\end{thebibliography}



\appendix 

\section*{Appendices}

\section{Proofs, examples and technical details}

\subsection{Proof of Theorem~\ref{thm:properties:xi}}

We will prove~\eqref{eq:xi:limit}, from which all the properties~\ref{xi:prop:range}--\ref{xi:prop:monotonicity} are straightforward to check.
First, the relationship~\eqref{eq:def:lower:tail:order}
yields, as $x\rightarrow \infty$, 
\begin{align}
\nonumber
\alpha_{C_1,C_2}(1/x) &= \log\left( \frac{C_2(1/x,\dots,1/x)}{C_1(1/x,\dots,1/x)}\right)\\
\nonumber &\sim \log \left(\frac{x^{-\kappa_2}\ell_2(1/x)}{x^{-\kappa_1}\ell_{1}(1/x)}\right)\\
\label{eq:asym:approx:nu}
& \sim (\kappa_{2}- \kappa_1) \log \left(\frac{1}{x}\right) + \log \ell_2(1/x) - \log \ell_1(1/x).
\end{align}
Notice that the functions $x \mapsto \ell_1(1/x)$ and $x \mapsto \ell_2(1/x)$ are slowly varying at $\infty$.
Therefore, regardless of $\kappa_1= \kappa_2$ or $\kappa_1\neq  \kappa_2$, we have that
\begin{align}\label{eq:xi:proof:case:1}
\lim_{u\downarrow 0}\frac{\alpha_{C_1,C_2}(u)}{\log u}= \kappa_2- \kappa_1,
\end{align}
by~\eqref{eq:asym:approx:nu} and Proposition 2.6 (i)~of~\cite{Resnick2007}.

Next, assume that $\kappa_1= \kappa_2$.
Except the two cases $\lambda_1=\lambda_2=0$ and $\lambda_1=\lambda_2=\infty$, we have that
\begin{align}\label{eq:xi:proof:case:2}
\lim_{v\downarrow 0}\alpha_{C_1,C_2}(v)= \log \left(\frac{\lambda_2}{\lambda_1}
\right),
\end{align}
which is zero if and only if $\lambda_1=\lambda_2$. 

Finally, if $\kappa_1\neq  \kappa_2$, then
\begin{align}\label{eq:xi:proof:case:3}
    \lim_{v\downarrow 0}\alpha_{C_1,C_2}(v) =
    \begin{cases}
        +\infty, & \text{ if }\kappa_1 > \kappa_2,\\
        -\infty,& \text{ if }\kappa_1 < \kappa_2.\\
    \end{cases}
\end{align} 

The asymptotic relationships~\eqref{eq:xi:proof:case:1}~\eqref{eq:xi:proof:case:2}~and~\eqref{eq:xi:proof:case:3}
are summarized as follows:
\begin{small}
\begin{align*}
\left(\lim_{u\downarrow 0}\frac{\alpha_{C_1,C_2}(u)}{\log (1/u)},\lim_{v\downarrow 0}\alpha_{C_1,C_2}(v)\right)=
\begin{cases}
(0,0),& \text{ if } 
\kappa_1=\kappa_2
\text{ and } 
\lambda_1=\lambda_2,\\
\left(0,\log\left(\lambda_2/\lambda_1\right)\right),& \text{ if }
\kappa_1=\kappa_2
 \text{ and }
 \lambda_1\neq \lambda_2
 ,\\
(\kappa_1-\kappa_2,\infty),
& \text{ if }
\kappa_1>\kappa_2,\\
(\kappa_1-\kappa_2,-\infty),
& \text{ if }
\kappa_1<\kappa_2,\\
\end{cases}
\end{align*}
\end{small}
and thus we obtain~\eqref{eq:xi:limit}.

\subsection{Proofs of Theorems~\ref{thm:consistency} and~\ref{thm:asymptotics}}

We first show the following lemma.
\begin{Lemma} \label{lem:moments}
For $u,v \in (0,1)$, it holds that
$$
\mathbb{E} \left[ \hat F_k(u) \right] = \frac{1}{n}+ F_k(u), \quad
\var \left( F_k (u) \right) = \frac1n F_k(u) \{ 1 - F_k(u) \},
$$
$$
\cov \left( \hat F_k (u), \hat F_k (v) \right) = \frac1n F_k (u\wedge v) \{ 1 - F_k (u \vee v) \},
$$
$$
\cov \left( \hat F_k(u), \hat F_{l}(v) \right) = \frac{1}{n} \left\{ H(u,v) - F_{k}(u) F_{l} (v) \right\},
$$
where $k,l=1,2$ with $k \neq l$.
\end{Lemma}

\begin{proof}
        The three summary measures
        $ \mathbb{E} [ \hat F_k(u) ]$, $
        \var ( \hat F_k (u) ) $ and $\cov ( \hat F_k(u), \hat F_{k}(v) )$ can be calculated in a similar manner as in Lemma 1 of \cite{Kato_etal2022} and therefore the proofs are omitted.
        The covariance of $\hat F_k(u)$ and $\hat F_{l}(v)$ can be expressed as
        \begin{align*}
            \cov \left( \hat F_k(u) \hat F_{l}(v) \right) = & \ \mathbb{E} \left[ \left(\hat F_k(u)-\frac{1}{n}\right)\left(\hat F_{l}(v)-\frac{1}{n}\right) \right] -  F_k(u)F_{l}(v).
   \end{align*}
   The first term of the right-hand side of this equation can be calculated as
   \begin{align*}
         \frac{1}{n^2} \sum_{i,j=1}^n \mathbb{E}& \left[ \id (
  M_i^{(k)}
  \leq u ) \id (
    M_j^{(l)}
  \leq v )\right] \\
  = & \ \frac{1}{n^2} \sum_{i=1}^n \mathbb{E} \left[ \id (
    M_i^{(k)}
  \leq u ) \id (
    M_i^{(l)}
  \leq v )\right]  + \frac{1}{n^2} \sum_{i \neq j} \mathbb{E} \left[ \id (
   M_i^{(k)}
 \leq u ) \id (
   M_j^{(l)}
 \leq v )\right] \\
  = & \ \frac{1}{n^2} \sum_{i=1}^n \mathbb{E} \left[ \id (
  M_{i}^{(k)}
 \leq u, 
  M_{i}^{(l)}
 \leq v )\right]   + \frac{1}{n^2} \sum_{i \neq j} \mathbb{E} \left[ \id (
  M_{i}^{(k)}
  \leq u ) \right] \mathbb{E} \left[ \id (
  M_j^{(l)}
  \leq v )\right] \\
= & \ \frac{1}{n^2} \left\{ \sum_{i=1}^n H(u,v) + \sum_{i \neq j} F_k (u) F_{l} (v)  \right\} \\
= & \ \frac{1}{n} \left\{ H(u,v) + (n-1) F_{k}(u) F_{l} (v) \right\}.
        \end{align*}
  Thus we have
  \begin{align*}
  \cov \left( \hat F_k(u), \hat F_{l}(v) \right) = & \ \frac{1}{n} \left\{ H(u,v) + (n-1) F_{k}(u) F_{l} (v) \right\}  - F_{k}(u) F_{l} (v) \\
  = & \ \frac{1}{n} \left\{ H(u,v) - F_{k}(u) F_{l} (v) \right\}.
  \end{align*}
\end{proof}

\begin{Remark}
Although the first three equations of this lemma are straightforward from Lemma 1 of \cite{Kato_etal2022}, the covariance $\cov ( \hat F_k(u), \hat F_{l}(v) )$ is essentially different from the covariance of the sample analogues of upper and lower tail probabilities \citep{Kato_etal2022}, which takes negative values in general.
\end{Remark}

Using this lemma, Theorem~\ref{thm:consistency} is straightforward from Lemma~\ref{lem:moments}, LLN and CMT~\citep[Theorem 2.3]{van1998asy}.
To prove Theorem~\ref{thm:asymptotics},
    let
    $$
    \hat{\bm{\beta}} =(\hat{\beta}_1,\hat{\beta}_2, \hat{\beta}_3, \hat{\beta}_4)^{\top} = \left( \hat F_1(u),\hat F_2(u),\hat F_1(v),\hat F_2(v) \right)^{\top},
    $$
    $$
    \bm{\beta} = (\beta_1,\beta_2,\beta_3,\beta_4)^{\top} = (F_1(u), F_2(u) , F_1(v), F_2(v) )^{\top},
    $$
    $$
    \bm{\Sigma_{\beta}} = (\sigma_{i j} )_{4\times 4}, \quad \sigma_{i j} = n \cdot \cov (\hat{\beta}_i , \hat{\beta}_j). 
    $$
    Then it follows from Lemma \ref{lem:moments}, LLN and CLT that
    \begin{equation}
    \sqrt{n} ( \hat{\bm{\beta}} - \bm{\beta})  \xrightarrow{d} \mathcal N ( \bm{0} , \bm{ \Sigma_{\beta}}) \quad \mbox{as } n \rightarrow \infty . \label{eq:beta_asy_normal}
    \end{equation}
    Next we define
    $$
    h (\bm{\beta}) = 
    \begin{pmatrix}
        \check{\xi}_u (\beta_2 / \beta_1) \\
        \check{\xi}_v (\beta_4 / \beta_3 )
    \end{pmatrix},
    $$
    where
    $$
    \check{\xi}_t (x) = (1-w) h_1 \left( \frac{\log(x)}{\log(1/t)} \right) + w h_2 \left( \log (x) \right).
    $$
 Then the delta-method implies that
    $$
    \sqrt{n} \left\{ h(\hat{\bm{\beta}}) - h(\bm{\beta})  \right\}  - \nabla h (\bm{\beta})^{\top}  \left\{ \sqrt{n} ( \hat{\bm{\beta}} - \bm{\beta}) \right\} \label{eq:delta_method} 
    $$
   weakly converges to zero as $n\rightarrow \infty$, where
    $$
    \nabla h (\bm{\beta}) = \left(
    \begin{array}{cc}
        \frac{\partial }{\partial \beta_1} \check{\xi}_u (\beta_2 / \beta_1) & \frac{\partial}{\partial \beta_1} \check{\xi}_v (\beta_4 / \beta_3) \\
        \vdots & \vdots \\
        \frac{\partial }{\partial \beta_4} \check{\xi}_u (\beta_2 / \beta_1) & \frac{\partial}{\partial \beta_4} \check{\xi}_v (\beta_4 / \beta_3) \\
    \end{array}
    \right) = \left(
    \begin{array}{cc}
        -a_w(u)/\beta_1 & 0 \\
        a_W(u)/\beta_2 & 0 \\
        0 & -a_w(v) / \beta_3 \\
        0 & a_w(v) / \beta_4
    \end{array}
    \right).
    $$
 This fact together with (\ref{eq:beta_asy_normal}) 
 implies that
$$
 \sqrt{n} \left\{ h(\hat{\bm{\beta}}) - h(\bm{\beta})  \right\} \xrightarrow{d} N \left( 0 , \nabla h(\bm{\beta})^{\top} \bm{ \Sigma_{\beta}} \nabla h(\bm{\beta}) \right) \quad \mbox{as } n \rightarrow \infty.
$$
    The asymptotic variance can be calculated as
    \begin{align}
        \lefteqn{ \nabla h(\bm{\beta})^{\top} \bm{\Sigma_{\beta}} \nabla h(\bm{\beta}) } \hspace{0.9cm} \nonumber \\
        = & \ a_w(u) a_w(v) \left(
         \begin{array}{cc}
            \frac{\sigma_{11}}{\beta_1^2} - \frac{2 \sigma_{12} }{ \beta_1 \beta_2 } + \frac{ \sigma_{22} }{ \beta_2^2}   & \frac{\sigma_{13}}{\beta_1 \beta_3} - \frac{\sigma_{23}}{ \beta_2 \beta_3 } - \frac{\sigma_{14} }{ \beta_1 \beta_4 } + \frac{ \sigma_{24} }{ \beta_2 \beta_4 } \\
            \frac{\sigma_{13}}{\beta_1 \beta_3} - \frac{\sigma_{23}}{ \beta_2 \beta_3 } - \frac{\sigma_{14} }{ \beta_1 \beta_4 } + \frac{ \sigma_{24} }{ \beta_2 \beta_4 } &  \frac{\sigma_{33}}{\beta_3^2} - \frac{2 \sigma_{34} }{ \beta_3 \beta_4 } + \frac{ \sigma_{44} }{ \beta_4^2}
        \end{array}
        \right)  \nonumber \\
        = & \ a_w(u) a_w(v) \Bigg(
        {\small \begin{array}{cc}
        \frac{1}{F_1(u)} + \frac{1}{F_2(u)} -  \frac{2 H(u,u)}{F_1(u) F_2(u)} & A(u,v) \\
           A(v,u)  & \frac{1}{F_1(v)} + \frac{1}{F_2(v)} -  \frac{2 H(v,v)}{F_1(v) F_2(v)}
        \end{array} }
        \Bigg) \nonumber \\
        = & \ \left(
        \begin{array}{cc}
            \sigma_{w}(u,u) & \sigma_{w}(u,v) \\
            \sigma_{w}(v,u) & \sigma_{w}(v,v)
        \end{array}
        \right). \label{eq:asy_var}
    \end{align}
 Therefore, it follows that $ (\mathbb{X}_n(u),\mathbb{X}_n(v)) $ converges weakly to the two-dimensional Gaussian distribution with mean vector $\bzero_2$ and the variance matrix~\eqref{eq:asy_var} as $n \rightarrow \infty$.
    Weak convergence of $ (\mathbb{X}_n(u_1), \ldots , \mathbb{X}_n(u_m))$ to an $m$-dimensional centered Gaussian distributions can be shown in a similar manner.

\subsection{Proof of Theorem~\ref{thm:consistency:hat:xi}}

We first prove that, for $v_n$ satisfying~\ref{item:A:v:n:consistency}, it holds that
\begin{align}\label{eq:consistency:lambda}
\frac{\hat F_l(v_n)}{v_n^{\kappa_l}} \parrow \lambda_l,\quad\text{ for }l=1,2.
\end{align}
To this end, let
$$
\frac{\hat F_l(v_n)}{v_n^{\kappa_l}} =\frac{1}{n v_n^{\kappa_l}}+\sum_{i=1}^n Y_{i,n}^{(l)},
$$
where
$$
Y_{i,n}^{(l)}=\frac{1}{n v_n^{\kappa_l}}\id\{
M_i^{(l)}\le v_n
\}.
$$
Then
$$
\mathbb{E}\left[\sum_{i=1}^n Y_{i,n}^{(l)}\right]=\sum_{i=1}^n \frac{1}{n} \frac{F_l(v_n)}{v_n^{\kappa_l}}\rightarrow \lambda_l\quad\text{ as }n\rightarrow \infty.
$$
Moreover,
\begin{align*}
\var(Y_{i,n}^{(l)})&=\mathbb{E}\left[\left(Y_{i,n}^{(l)}\right)^2\right] - \left(
\mathbb{E}\left[Y_{i,n}^{(l)}\right]
\right)^2\\
&= \frac{F_l(v_n)}{n^2 (v_n^{\kappa_l})^2} - \frac{F_l^2(v_n)}{n^2 (v_n^{\kappa_l})^2}.
\end{align*}
Since $Y_{i,n}^{(l)}$, $i=1,\dots,n$, are independent, we have that
\begin{align*}
\var\left(
\sum_{i=1}^n Y_{i,n}^{(l)}
\right)
&=
\sum_{i=1}^n \var\left(Y_{i,n}^{(l)}
\right)\\
&= \frac{1}{nv_n^{\kappa_l}}\frac{F_l(v_n)}{v_n^{\kappa_l}}-\frac{1}{n}\left\{
\frac{F_l(v_n)}{v_n^{\kappa_l}}
\right\}^2\\
&\rightarrow 0\quad \text{ as }n\rightarrow\infty.
\end{align*}
Therefore, we obtain~\eqref{eq:consistency:lambda} by LLN for triangular arrays~\citep{durrett2019probability}.

Note that
\begin{align*}
\hat \alpha(v_n)=\log\left(
\frac{\hat F_2(v_n)}{v_n^{\kappa_2}}
\right)
-
\log\left(
\frac{\hat F_1(v_n)}{v_n^{\kappa_1}}
\right)
+(\kappa_2-\kappa_1)\log v_n,
\end{align*}
and $\log v_n \rightarrow -\infty$ as $n\rightarrow \infty$.
Therefore, if $\kappa_1=\kappa_2$, we have from CMT and the Slutsky's theorem that $h_2(\hat \alpha(v_n))\parrow h_2(\log(\lambda_2/\lambda_1))$.
We next prove that $h_2(\hat \alpha(v_n))\parrow 1$ if $\kappa_1>\kappa_2$.
Let
\begin{align}\label{eq:delta:n}
\Delta_n = \left\{
\log\left(
\frac{\hat F_2(v_n)}{v_n^{\kappa_2}}
\right)
-
\log\left(
\frac{\hat F_1(v_n)}{v_n^{\kappa_1}}
\right)\right\} -\log \frac{\lambda_2}{\lambda_1}.
\end{align}
Then $\Delta_n \parrow 0$ by~\eqref{eq:consistency:lambda} and thus, for every  $\epsilon>0$ and $\delta>0$, there exists $n_0\in\mathbb{N}$ such that $\mathbb{P}(|\Delta_n|<\delta)>1-\epsilon$ for every $n$ satisfying $n\ge n_0$.
Note that $\mathbb{P}(|h_2(\hat \alpha(v_n))|\le 1)=1$ for every $n$.
For every $\tilde \delta>0$, we have that
\begin{align*}
    \mathbb{P}\left(h_2(\hat \alpha(v_n))>1-\tilde \delta\right)
    &=\mathbb{P}\left(\hat \alpha(v_n)>h_2^{-1}\left(1-\tilde \delta\right)\right)\\
    &=\mathbb{P}\left(\Delta_n > h_2^{-1}\left(1-\tilde \delta\right)- \log(\lambda_2/\lambda_1) -(\kappa_2-\kappa_1)\log v_n\right).
\end{align*}
Since $(\kappa_2-\kappa_1)\log v_n\rightarrow \infty$ as $n\rightarrow\infty$, we can take sufficiently large $\tilde n_0\in\mathbb{N}$ such that $\tilde n_0 > n_0$ and $-\delta > h_2^{-1}\left(1-\tilde \delta\right)- \log(\lambda_2/\lambda_1) -(\kappa_2-\kappa_1)\log v_n$.
Thus, for every $n$ such that $n>\tilde n_0$, we have that $\mathbb{P}\left(h_2(\hat \alpha(v_n))>1-\tilde \delta\right)>1-\epsilon$, which proves that $h_2(\hat \alpha(v_n))\parrow 1$.
By symmetry, it follows analogously that $h_2(\hat \alpha(v_n))\parrow -1$ if $\kappa_1<\kappa_2$.
In summary, it holds that
\begin{align*}
h_2(\hat \alpha(v_n))\parrow\begin{cases}
     h_2(\log(\lambda_2/\lambda_1)),& \text{ if }\kappa_1=\kappa_2,\\
    1,&\text{ if }\kappa_1>\kappa_2,\\
    -1,&\text{ if }\kappa_1<\kappa_2.\\
\end{cases}
\end{align*}
Together with~\ref{item:A:tail:index:consistency}, we conclude that $\hat \xi_w\parrow \xi_w$ by CMT and the Slutsky's theorem.

\subsection{Proof of Theorem~\ref{thm:clt:hat:xi}}

We begin with the following lemma on the tail order parameter.

\begin{Lemma}\label{lem:lambda:clt}
Let $l\in \{1,2\}$ be such that $\kappa_l=\kappa_{+}$.
Under the assumptions in Theorem~\ref{thm:clt:hat:xi}, it holds that
\begin{align}\label{eq:clt:lambda}
\sqrt{m_n}\left\{
\frac{\hat F_l(v_n)}{v_n^{\kappa_l}} - \lambda_l
\right\}\darrow \mathcal N(0,\tau_l\lambda_l),\quad \text{ as }n\rightarrow\infty,
\end{align}
and thus
\begin{align}\label{eq:clt:log:lambda}
\sqrt{m_n}\left\{
\log \frac{\hat F_l(v_n)}{v_n^{\kappa_l}} - \log \lambda_l
\right\}\darrow \mathcal N\left(0,\frac{\tau_l}{\lambda_l}\right),\quad\text{ as } n\rightarrow\infty.
\end{align}
\end{Lemma}

\begin{proof}[Proof of Lemma~\ref{lem:lambda:clt}]
We will prove~\eqref{eq:clt:lambda} since \eqref{eq:clt:log:lambda} is an immediate consequence from~\eqref{eq:clt:lambda} by the delta-method.
When~\ref{item:A:F:l} holds, the weak limit of the LHS of~\eqref{eq:clt:lambda} coincides with that of 
$$
\sqrt{m_n}\left\{
\frac{\hat F_l(v_n)}{v_n^{\kappa_l}} - \frac{F_l(v_n)}{v_n^{\kappa_l}} \right\}=\frac{\sqrt{m_n}}{nv_n^{\kappa_l}}+\sum_{i=1}^n Z_{i,n}^{(l)},
$$
where
$$
Z_{i,n}^{(l)}=\frac{\sqrt{m_n}}{n v_n^{\kappa_l}} \left\{
 \id\{ M_i^{(l)}\le v_n\}- 
F_l(v_n) 
 \right\}.
$$
By~\ref{item:mn:un:lim} and~\ref{item:A:tau}, the first term $\sqrt{m_n}/nv_n^{\kappa_l}$ converges to $0$ as $n\rightarrow\infty$.
We will apply the Lindeberg-Feller theorem~\citep{billingsley2017probability,durrett2019probability} to the second term $\sum_{i=1}^n Z_{i,n}^{(l)}$.
We can first check that $Z_{i,n}^{(l)}$, $i=1,\dots,n$, are independent, and that $\mathbb{E}[Z_{i,n}^{(l)}]=0$.
Moreover, since $m_n/n= \tau_l u_n^{\kappa_l}\rightarrow 0$ by~\ref{item:A:tau}, we have that
\begin{align*}
s_n^2:&=\var\left( \sum_{i=1}^n Z_{i,n}^{(l)}
\right)=\sum_{i=1}^n \mathbb{E}[(Z_{i,n}^{(l)})^2]\\
&=
n \frac{m_n}{n^2 v_n^{2\kappa_l}}\left\{
F_l(v_n)-F_l^2(v_n)\right\}\\
&=
\frac{m_n}{n v_n^{\kappa_l}}\frac{F_l(v_n)}{v_n^{\kappa_l}}-\frac{m_n}{n}
\left\{
\frac{F_l(v_n)}{v_n^{\kappa_l}}
\right\}^2\\
&\rightarrow\tau_l  \lambda_l.
\end{align*}
Finally, we check the Lindeberg condition.
Since 
\begin{align*}
(Z_{i,n}^{(l)})^2\le \frac{m_n}{n^2 v_n^{2\kappa_l}}=\frac{\tau_l^2}{m_n},
\end{align*}
we have that, for every $\epsilon>0$,
\begin{align*}
\sum_{i=1}^n \mathbb{E}\left[
\left(Z_{i,n}^{(l)}\right)^2 
\id\{|Z_{i,n}^{(l)}|\ge \epsilon s_n\}\right]
&\le
\frac{\tau_l^2}{m_n}\sum_{i=1}^n \mathbb{P}\left(
|Z_{i,n}^{(l)}|^2\ge \epsilon^2 s_n^2
\right)\\
&\le
\frac{\tau_l^2}{m_n}\sum_{i=1}^n  \frac{\mathbb{E}\left[|Z_{i,n}^{(l)}|^2\right]}{\epsilon^2 s_n^2}\\
&=\frac{\tau_l^2}{\epsilon^2 m_n}\rightarrow 0,
\end{align*}
where the second inequality comes from the Markov inequality.
Therefore, we have the desired result~\eqref{eq:clt:lambda}.
\end{proof}

We next consider the asymptotic behavior of $h_2(\hat \alpha(v_n))$ for the case when $\kappa_1=\kappa_2$.

\begin{Lemma}\label{lem:clt:equal:kappa}
Assume that $\kappa_1=\kappa_2=:\kappa$ in addition to the assumptions in Theorem~\ref{thm:clt:hat:xi}.
Then we have that
\begin{align*}
\sqrt{m_n}
\left\{
h_2(\hat \alpha(v_n)) - h_2\left(\log \frac{\lambda_2}{\lambda_1}\right)
\right\}
\darrow\mathcal N(0,\sigma^2)
\quad \text{ as } n\rightarrow\infty,
\end{align*}
where $\sigma^2$ is given in~\eqref{eq:sigma:clt:h2}.
\end{Lemma}

\begin{proof}
If $\kappa_1=\kappa_2=\kappa$, then Lemma~\ref{lem:lambda:clt} applies to $l=1,2$.
By~\ref{item:A5:independence}, it holds that
\begin{align}\label{eq:clt:joint:lambda}
\sqrt{m_n}\left\{
\left(
\frac{\hat F_1(v_n)}{v_n^{\kappa}},
\frac{\hat F_2(v_n)}{v_n^{\kappa}} 
\right)^{\top} - 
(\lambda_1,\lambda_2)^\top
\right\}\darrow \mathcal N_2\left(\bzero_2, 
\begin{pmatrix}
\tau\lambda_1& 0\\
0& \tau\lambda_2\\
\end{pmatrix}
\right).
\end{align}
Define the function $g:\mathbb{R}_{+}\times \mathbb{R}_{+}\rightarrow \mathbb{R}$ by $
(x,y)\mapsto g\left(\log(y/x)\right)$.
Then
\begin{align*}
    \nabla g(\lambda_1,\lambda_2)&=\left(\frac{\partial}{\partial x}g(x,y),\frac{\partial}{\partial y}g(x,y)\right)^\top\\
    &=\left(
    -\frac{1}{x} h_2'\left(\log \frac{y}{x}\right),
    \frac{1}{y} h_2'\left(\log \frac{y}{x}\right)
\right)^\top.
    \end{align*}
By the delta-method, it holds from~\eqref{eq:clt:joint:lambda} that
\begin{align*}
\sqrt{m_n}
\left\{
h_2(\hat \alpha(v_n)) - h_2\left(\log \frac{\lambda_2}{\lambda_1}\right)
\right\}
&=
\sqrt{m_n}
\left\{
g\left(
\frac{\hat F_1(v_n)}{v_n^{\kappa}},\frac{\hat F_2(v_n)}{v_n^{\kappa}}
\right)
- g(\lambda_1,\lambda_2)
\right\}\\
&\darrow \mathcal N(0,\sigma^2),\quad \text{ as }n\rightarrow\infty,
\end{align*}
where 
\begin{align*}
\sigma^2 &= \nabla g(\lambda_1,\lambda_2)^{\top}
\begin{pmatrix}
\tau \lambda_1& 0\\
0& \tau \lambda_2\\
\end{pmatrix}
\nabla g(\lambda_1,\lambda_2)\\
&=\tau \left\{
h_2'\left(\log \frac{\lambda_2}{\lambda_1}\right)
\right\}^2
\left( -\frac{1}{\lambda_1},\frac{1}{\lambda_2}\right)^{\top}
\begin{pmatrix}
 \lambda_1& 0\\
0& \lambda_2\\
\end{pmatrix}
\left( -\frac{1}{\lambda_1},\frac{1}{\lambda_2}\right)
\\
&=\tau\left\{h_2'\left(
\log\frac{\lambda_2}{\lambda_1}
\right)\right\}^2
\left(
\frac{1}{\lambda_1}+\frac{1}{\lambda_2}
\right).
\end{align*}
\end{proof}

Our last lemma is on the asymptotic behavior of $h_2(\hat \alpha(v_n))$ for the case when $\kappa_1\neq \kappa_2$.

\begin{Lemma}\label{lem:h2:vanish}
    Assume that $\kappa_1\neq \kappa_2$ in addition to the assumptions in Theorem~\ref{thm:clt:hat:xi}.
Then we have that
\begin{align*}
\sqrt{m_n}
\left\{
h_2(\hat \alpha(v_n)) - \operatorname{sign}(\kappa_1>\kappa_2)
\right\}
\parrow0.
\end{align*}
\end{Lemma}

\begin{proof}
We first assume that $\kappa_1>\kappa_2$.
Let $\Delta_n$ be as defined in~\eqref{eq:delta:n}.
When~\ref{item:A:tau} holds, then~\ref{item:A:v:n:consistency} holds for $\kappa_l$, $l=1,2$, and thus we have $\Delta_n \parrow 0$ by Theorem~\ref{thm:consistency:hat:xi}.
Therefore, for every $\epsilon>0$ and $\delta>0$, there exists $n_0\in\mathbb{N}$ such that $\mathbb{P}(|\Delta_n|<\delta)>1-\epsilon$ for every $n$ with $n\ge n_0$.
For every $\tilde \delta>0$, we have that
\begin{align*}
    \mathbb{P}\left(\sqrt{m_n}\{1-h_2(\hat \alpha(v_n))\}>\tilde \delta\right)
    &=\mathbb{P}\left(\hat \alpha(v_n)<h_2^{-1}\left(1-\tilde \delta/\sqrt{m_n}\right)\right)\\
    & \le \mathbb{P}\left(\hat \alpha(v_n)<x^\ast\right)\\
    & =\mathbb{P}\left(\Delta_n <x^\ast +\log(\lambda_2/\lambda_1) -(\kappa_2-\kappa_1)\log v_n\right)
\end{align*}
Since $(\kappa_2-\kappa_1)\log v_n\rightarrow \infty$ as $n\rightarrow\infty$, we can take sufficiently large $\tilde n_0\in\mathbb{N}$ such that $\tilde n_0 > n_0$ and $
-\delta > x^\ast +\log(\lambda_2/\lambda_1) -(\kappa_2-\kappa_1)\log v_n$.
Thus, for every $n$ with $n>\tilde n_0$, we have that $\mathbb{P}\left(\Delta_n <x^\ast +\log(\lambda_2/\lambda_1) -(\kappa_2-\kappa_1)\log v_n\right)<\epsilon$, which proves that $\sqrt{m_n}\{h_2(\hat \alpha(v_n))-1\}\parrow 0$.
By symmetry, it follows analogously that $\sqrt{m_n}\{h_2(\hat \alpha(v_n))+1\}\parrow 0$ when $\kappa_1<\kappa_2$.
\end{proof}

Consider Case~\ref{item:thm:case:I}.
If $w=0$, then~\eqref{eq:clt:case:I} is an immediate consequence from~\ref{item:A:tail:index:clt} and CMT.
Therefore, assume that $m_n/\tilde m_n\rightarrow \infty$ and $w\in(0,1]$.
If $\kappa_1=\kappa_2$, then Lemma~\ref{lem:clt:equal:kappa} implies that
\begin{align*}
    \sqrt{\tilde m_n}\left\{\hat \xi_w - \xi_w\right\}
    &=(1-w)\sqrt{\tilde m_n}\left(h_1(\hat \kappa_1-\hat \kappa_2)-h_1(\kappa_1-\kappa_2)\right)\\
    &\qquad +
    w\sqrt{\tilde m_n/m_n}\sqrt{m_n}\left\{h_2(\hat \alpha(v_n)) - h_2(\log(\lambda_2/\lambda_1))\right\}\\
    &
    \darrow \mathcal N(0,(1-w)^2\,\tilde \sigma^2)\quad \text{ as }n\rightarrow \infty.
\end{align*}
If $\kappa_1\neq \kappa_2$, then Lemma~\ref{lem:h2:vanish} leads to the same conclusion $\sqrt{\tilde m_n}\left\{\hat \xi_w - \xi_w\right\}\darrow \mathcal N(0,(1-w)^2\,\tilde \sigma^2)$ as $n\rightarrow \infty$.

We next consider Case~\ref{item:thm:case:II}.
If $w=1$, then~\eqref{eq:clt:case:II} is an immediate consequence from~Lemma~\ref{lem:clt:equal:kappa} and CMT.
Therefore, suppose that $\tilde m_n/m_n\rightarrow \infty$ and $w\in[0,1)$. 
Since $\kappa_1=\kappa_2$, Lemma~\ref{lem:clt:equal:kappa} implies that
\begin{align*}
    \sqrt{ m_n}\left\{\hat \xi_w - \xi_w\right\}
    &=(1-w)\sqrt{m_n/{\tilde m_n}}\sqrt{\tilde m_n}\left(h_1(\hat \kappa_1-\hat \kappa_2)-h_1(\kappa_1-\kappa_2)\right)\\
    &\qquad +
    w\sqrt{m_n}\left\{h_2(\hat \alpha(v_n)) - h_2(\log(\lambda_2/\lambda_1))\right\}\\
    &
    \darrow \mathcal N(0,w^2\,\sigma^2)\quad \text{ as }n\rightarrow \infty,
\end{align*}
where $\sigma^2$ is given in~\eqref{eq:sigma:clt:h2}.

\subsection{Illustration of Theorems~\ref{thm:consistency:hat:xi} and~\ref{thm:clt:hat:xi}}\label{sec:example}

For illustration, we consider the bivariate case $d=2$.
Suppose that bivariate copulas $C_l$, $l=1,2$, admit the following expansions:
\begin{align}\label{eq:expansion}
C_l(u,u)=\lambda_l u^{\kappa_l} + \theta_l u^{\nu_l}+ o(u^{\nu_l}),\quad u\rightarrow 0,
\end{align}
where $\lambda_l>0$, $\nu_l>\kappa_l>0$ and $\theta_l \in\mathbb{R}$.
This expansion is satisfied, for example, by the FGM copula
\begin{align}\label{eq:fgm}
C(u,v;\delta)=uv\{1+\delta(1-u)(1-v)\},\quad u,v \in [0,1],\,\delta \in [-1,1],
\end{align}
where
$$
C(u,u;\delta)=(1+\delta)u^2 -2\delta u^3 + \delta u^4,\quad u \in  [0,1].
$$
Suppose that $\kappa_{+}=\kappa_2$.
In addition, we assume the mild condition that $\kappa_2 +3\nu_1 >\nu_2$.
Let $\tau=1$ in~\ref{item:A:tau},
\begin{align}\label{eq:vn:mn}
v_n = n^{-\frac{1}{\kappa_2+2\nu_2}-\epsilon}\quad\text{and}\quad
m_n = n^{1-\kappa_2\left(\epsilon + \frac{1}{\kappa_2 + 2\nu_2}\right)},
\end{align}
where $\epsilon\in\mathbb{R}$ is such that
\begin{align}\label{eq:choice:eps}
\max\left(0,\frac{1}{\kappa_2+2\nu_1}-\frac{1}{\kappa_2+2\nu_2}\right) 
< \epsilon < \frac{2\nu_2}{\kappa_2(\kappa_2 + 2\nu_2)}.
\end{align}
Note that such an $\epsilon$ exists if $\nu_1 \ge \nu_2$ or $\nu_1<\nu_2<\kappa_2 +3\nu_1$.
For $(m_n,v_n)$ in~\eqref{eq:vn:mn}, the condition~\ref{item:A:tau} is satisfied. 
In addition, conditions in~\ref{item:A:v:n:consistency} and~\ref{item:mn:un:lim} are satisfied since, as $n\rightarrow\infty$, we have that
$v_n \rightarrow 0$, $m_n=n v_n^{\kappa_2} \rightarrow \infty$ and $nv_n^{\kappa_1}=nv_n^{\kappa_2}v_n^{\kappa_1-\kappa_2}\rightarrow\infty$.
Under~\eqref{eq:expansion}, the conditions in~\ref{item:A:F:l} reduce to $\sqrt{m_n}v_n^{\nu_l}\rightarrow 0$ for $l=1,2$, which are induced by $n v_n^{\kappa_2+2\nu_l}\rightarrow0$ for $l=1,2$.
For $l=2$, this condition is trivial since $
n v_n^{\kappa_2+2\nu_2}= n^{-\epsilon (\kappa_2+2\nu_2)}\rightarrow 0.
$
For $l=1$, the condition is equivalent to
\begin{align*}
1-\frac{\kappa_2+2\nu_1}{\kappa_2+2\nu_2}-(\kappa_2+2\nu_1)\epsilon<0,
\end{align*}
and is satisfied from~\eqref{eq:choice:eps}.

Regarding the conditions~\ref{item:A:tail:index:consistency} and~\ref{item:A:tail:index:clt}, we consider the situation that the tail orders $\kappa_l$, $l=1,2$, are estimated by the following \emph{Hill estimators}~\citep{hill1975simple}: 
\begin{align}\label{eq:hill:kappa}
    \hat \kappa_{l}=\left(
        \frac{1}{k_l}\sum_{j=1}^{k_l} \log \frac{\hat M_{k_l+1,n}^{(l)}}{{\hat M_{j,n}^{(l)}}}
    \right)^{-1},
\end{align}
where $k_l=k_l(n)\in \mathbb{N}$ is such that $k_l\rightarrow \infty$ and $k_l/n\rightarrow \infty$ as $n\rightarrow \infty$.
Moreover, $\hat M_{1,n}^{(l)}\le \hat M_{2,n}^{(l)}\le \cdots \le \hat M_{n,n}^{(l)}$ are the order statistics of $\hat M_1^{(l)},\dots,\hat M_n^{(l)}$, where $\hat M_{i}^{(l)}=\max\left\{
\hat{U}_{ij}^{(l)},j=1,\dots,d
\right\}$, $i=1,\dots,n$, with $\hat{U}_{ij}^{(l)}$ being the rank of $U_{ij}^{(l)}$ among $U_{1j}^{(l)},\dots,U_{nj}^{(l)}$.

The reciprocal of~\eqref{eq:hill:kappa} is studied in~\cite{draisma2004bivariate}.
The following statements are immediate consequences of Theorems~2.1 and~2.2 of~\cite{draisma2004bivariate} by the delta-method and CMT.

\begin{Theorem}\label{thm:hill}
In addition to~\ref{item:A4:iid}, assume the following conditions for $l=1,2$.
    \begin{enumerate}[label=(a.\arabic*)]\setcounter{enumi}{0}
 \item\label{item:a1:second:rv}
The limit:
\begin{align}\label{eq:second:order:variation}
    \lim_{t\downarrow0}\left.\left\{
\frac{C_l(tu,tv)}{q_l(t)} - p_l(u,v)
    \right\}\right/{\tilde q_l(t)}=:\tilde p_l(u,v)
\end{align}
exists for all $u,v\ge 0$ with $u+v >0$, where $q_l(t)=C_l(t,t)$, $p_l$ is such that $p_l(1,1)=1$, $\tilde q_l$ is a positive function with $\lim_{t\downarrow0}\tilde q_l(t)=0$, and $\tilde p_l$ is a function which is neither constant nor a multiple of $p_l$.
\item\label{item:a2:unoform:conv}
The convergence in~\eqref{eq:second:order:variation} is uniform on $\{(u,v)\in[0,\infty)^2\mid u^2+v^2=1\}$.
\item\label{item:a3:derivatives} The function $p_l$ has first-order derivatives $\partial_1 p_l(u,v)=\partial p_l(u,v)/\partial u$ and $\partial_2 p_l(u,v)=\partial p_l(u,v)/\partial v$.
\item\label{item:a4:tdc} The limit $\upsilon_l:=\lim_{t \downarrow0}q_l(t)/t$ exists.
\item\label{item:a5:quantile} $\sqrt{k_l}\,\tilde q_l(q_l^{-1}(k_l/n))\rightarrow0$ as $n\rightarrow\infty$.
\end{enumerate}
Then the following statements holds for $l=1,2$.
\begin{enumerate}[label=(\Roman*)]\setcounter{enumi}{0}
 \item\label{item:clt:draisma}
 $\hat \kappa_l \parrow \kappa_l$ and
$\sqrt{k_l}\,(\hat \kappa_l - \kappa_l)\darrow \mathcal N(0,\tilde \sigma_l^2)$ as $n\rightarrow\infty$, where
\begin{align*}
    \sigma_l^2 = \kappa_l^2 (1-\upsilon_l)\left\{
    1-2\upsilon_l \partial_1 p_l(1,1)\partial_2 p_l(1,1)
    \right\}.
\end{align*}
\item\label{item:sigma:consistency:draisma} 
Let 
\begin{align*}
\hat \sigma_l^2 = \hat \kappa_l \left(
1-\frac{k_l}{{\hat M}_{k_l+1,n}^{(l)}}
\right)
\left\{
1-2\frac{k_l}{\sqrt{{\hat M}_{k_l+1,n}^{(l)}}}
\left(
\frac{{\hat M}_{k_l+1,n}^{(l)}}{{\hat M}_{k_l+1,n}^{(l)\,[1]}}-1
\right)
\left(
\frac{{\hat M}_{k_l+1,n}^{(l)}}{{\hat M}_{k_l+1,n}^{(l)\,[2]}}-1
\right)
\right\},
\end{align*}
where  $\hat M_{k_l+1,n}^{(l)\,[1]}$ is the $(k_l+1)$th order statistic of 
\begin{align*}\max\left(
\frac{\hat U_{i1}^{(l)}}{1+\left({\hat M}_{k_l+1,n}^{(l)}\right)^{-1/4}},\hat U_{i2}
    \right), \quad i=1,\dots,n,
    \end{align*}
    and  $\hat M_{k_l+1,n}^{(l)\,[2]}$ is that of 
\begin{align*}\max\left(
\hat U_{i1},
\frac{\hat U_{i2}^{(l)}}{1+\left({\hat M}_{k_l+1,n}^{(l)}\right)^{-1/4}}
    \right), \quad i=1,\dots,n.
    \end{align*}
Then $\hat \sigma_l^2\parrow \sigma_l^2$.
\end{enumerate}
\end{Theorem}

Note that $\tilde q_l\in \operatorname{RV}_0(\tilde \tau_l)$ for some $\tilde \tau_l \ge 0$ under~\eqref{eq:second:order:variation}. 
Moreover, the condition~\ref{item:a5:quantile} is satisfied if  
$$
k_l=O\left(n^{\frac{2\tilde \tau_l}{2\tilde \tau_l+\kappa_l}-\epsilon_l}\right),\quad
\text{where }0<\epsilon_l <\frac{2\tilde \tau_l}{2\tilde \tau_l+\kappa_l}.
$$
Therefore, one can always take $\epsilon_1$ and $\epsilon_2$ such that $\tilde m_n=k_1=k_2$ and $m_n/\tilde m_n\rightarrow\infty$ as $n\rightarrow\infty$, under which Case~\ref{item:thm:case:I} occurs.
Note that the condition~\ref{item:A:tail:index:clt} can be shown straightforwardly from Theorem~\ref{thm:hill} by the delta-method and independence between $\hat \kappa_1$ and $\hat \kappa_2$.

As a concrete example, the FGM copula~\eqref{eq:fgm} satisfies the above assumption with $q(t)=(1+\delta)t^2$, $p_l(u,v)=uv$, $\tilde q(t)=-2\delta t/(1+\delta)$, $\tilde \tau=1$ and $\tilde p_l(u,v)=(u^2v + u v^2)/2$.
For $l=1,2$, suppose that $C_l$ is an FGM copula with parameter $\delta_l$.
By taking $\tilde m_n=k_1=k_2=n^{1/2-\tilde \epsilon}$ for $0<\tilde \epsilon < 1/2$, the conditions in~\ref{item:A:tail:index:consistency},~\ref{item:mn:un:lim} and~\ref{item:A:tail:index:clt} are satisfied.
In this case, we also have that $m_n=n^{3/4-2\epsilon}$, where $0<\epsilon < 3/8$.
Therefore, we have that
$
m_n/\tilde m_n=n^{1/4-2\epsilon+\tilde \epsilon}$, which can be $\infty$, $1$ or $0$ depending on the choice of $(\epsilon,\tilde \epsilon)$.

\end{document}